\documentclass[journal,twoside,web]{ieeecolor}

\usepackage{generic}
\usepackage{cite}
\usepackage{amsmath,amssymb,amsfonts}
\usepackage{algorithmic}
\usepackage{graphicx}
\usepackage{algorithm,algorithmic}
\usepackage[hidelinks]{hyperref}
\usepackage{textcomp}

\usepackage{mathrsfs}
\usepackage{amssymb}
\usepackage{color}
\usepackage{float}
\usepackage{graphicx}
\usepackage{booktabs}
\usepackage{longtable,tabularx}
\usepackage{nicefrac}
\usepackage{mathtools}
\usepackage{bm}
\usepackage{multirow}
\usepackage{apptools}
\usepackage{savesym}
\usepackage{babel}
\usepackage{ulem}
\let\labelindent\relax
\usepackage{enumitem}

\pdfoutput=1

\newcommand{\eigmax}{\boldsymbol{\lambda_{\rm max}}}
\newcommand{\eigmin}{\boldsymbol{\lambda_{\rm min}}}
\newcommand{\svdmax}{\boldsymbol{\sigma_{\rm max}}}
\newcommand{\svdmin}{\boldsymbol{\sigma_{\rm min}}}

\DeclareMathOperator*{\argmin}{arg\,min}

\newcommand{\rmT}{{\rm T}}

\newcommand{\BBN}{{\mathbb N}}

\newcommand{\BBR}{{\mathbb R}}

\newcommand{\SD}{{\mathcal D}}

\newcommand{\SW}{{\mathcal W}}

\newtheorem{theo}{Theorem}

\newtheorem{cor}{Corollary}
\newtheorem{defin}{Definition}

\newtheorem{lem}{Lemma}

\def\BibTeX{{\rm B\kern-.05em{\sc i\kern-.025em b}\kern-.08em
    T\kern-.1667em\lower.7ex\hbox{E}\kern-.125emX}}
\makeatletter
\newcommand{\myitem}[1]{%
\item[#1]\protected@edef\@currentlabel{#1}}
\makeatother

\makeatletter
\newcommand\footnoteref[1]{\protected@xdef\@thefnmark{\ref{#1}}\@footnotemark}
\makeatother

\begin{document}
\title{Generalized Forgetting Recursive Least Squares: Stability and Robustness Guarantees}
\author{Brian Lai and Dennis S. Bernstein, \IEEEmembership{Life Fellow, IEEE}
\thanks{Manuscript received 19 July 2023. This work was supported
by the NSF Graduate Research Fellowship under Grant DGE 1841052. (Corresponding author: Brian Lai.)}
\thanks{The authors are with the Department of Aerospace Engineering,
University of Michigan, Ann Arbor, MI 48109 USA (e-mail: brianlai@
umich.edu; dsbaero@umich.edu). }}

\maketitle

\begin{abstract}
This work presents generalized forgetting recursive least squares (GF-RLS), a generalization of recursive least squares (RLS) that encompasses many extensions of RLS as special cases. 
First, sufficient conditions are presented for the 1) Lyapunov stability, 2) uniform Lyapunov stability, 3) global asymptotic stability, and 4) global uniform exponential stability of parameter estimation error in GF-RLS when estimating fixed parameters without noise.
Second, robustness guarantees are derived for the estimation of time-varying parameters in the presence of measurement noise and regressor noise.
These robustness guarantees are presented in terms of global uniform ultimate boundedness of the parameter estimation error.
A specialization of this result gives a bound to the asymptotic bias of least squares estimators in the errors-in-variables problem.
Lastly, a survey is presented to show how GF-RLS can be used to analyze various extensions of RLS from the literature.

\end{abstract}

\begin{IEEEkeywords}
%
errors-in-variables, identification, recursive least squares, robustness, stability analysis
\end{IEEEkeywords}

\section{Introduction}
\label{sec:introduction}
\IEEEPARstart{R}{ecursive} least squares (RLS) is a foundational algorithm in systems and control theory for the online identification of fixed parameters \cite{ljung1983theory,aastrom2013adaptive,islam2019recursive}. 
A property of RLS is the eigenvalues of the covariance matrix are monotonically decreasing over time and may become arbitrarily small \cite[subsection 2.3.2]{sastry1990adaptive}, \cite{goel2020recursive}, resulting in eventual sluggish adaptation and inability to track time-varying parameters \cite{ortega2020modified,Salgado1988modified}.
Numerous extensions of RLS have been developed to improve identification of time-varying parameters, including exponential forgetting \cite{johnstone1982exponential,islam2019recursive}, variable-rate forgetting \cite{bruce2020convergence,fortescue1981implementation,Dasgupta1987Asymptotically,paleologu2008robust,Hung2005Gradient,mohseni2022recursive}, directional forgetting \cite{cao2000directional,kulhavy1984tracking,bittanti1990convergence}, resetting \cite{Salgado1988modified,shin2020new,Lai2023Exponential}, and multiple forgetting \cite{vahidi2005recursive}, among others.
Hence, we use the general term \textit{forgetting} to describe the processes in extensions of RLS which break the monotonicity of the covariance matrix.

Furthermore, several general frameworks have been developed which include extensions of RLS as special cases, for example \cite{parkum1992recursive} in discrete-time and \cite{shaferman2021continuous} in continuous-time. 
The recent work of \cite{bin2022generalized} develops a much more general framework of recursive estimators which contains RLS extensions as a special case.
These frameworks help to unify various RLS extensions and provide overarching analysis.
In discussing RLS extensions, we highlight three important points: 1) cost function, 2) stability, and 3) robustness.

\subsubsection{Cost Function}
The RLS update equations are derived as a recursive method to find the minimizer of a least-squares cost function \cite{islam2019recursive}.
While some RLS extensions can be derived from modified least-squares cost functions (e.g. exponential forgetting \cite{islam2019recursive} and variable-rate forgetting \cite{bruce2020convergence}), many have been developed as ad-hoc modifications to the RLS update equations, without an associated cost function (e.g. \cite{cao2000directional,kulhavy1984tracking,bittanti1990convergence,
Salgado1988modified,Lai2023Exponential,vahidi2005recursive,
paleologu2008robust,johnstone1982exponential,Dasgupta1987Asymptotically}). 
Therefore, there is an interest in developing a cost function from which extensions of RLS can be derived. 
In \cite{shaferman2021continuous}, a continuous-time cost functional is presented, from which continuous-time least-squares algorithms are derived. 
However, the least-squares algorithms derived from the cost function in \cite{shaferman2021continuous} are continuous-time versions of RLS extensions, not the original RLS extensions developed in discrete-time.

\subsubsection{Stability}
Many RLS extensions give conditions which guarantee stability of parameter estimation error to zero when estimating constant parameters \cite{bruce2020convergence,bittanti1990convergence,johnstone1982exponential,shin2020new}. General frameworks \cite{parkum1992recursive}, \cite{bin2022generalized}, and \cite{shaferman2021continuous} all present stability guarantees which apply to various RLS extensions. 
The stability analyses in \cite{parkum1992recursive} and \cite{shaferman2021continuous} consider RLS extensions with scalar measurements and present exponential stability guarantees for constant-parameter estimation. 
While the stability analysis in \cite{bin2022generalized} encompasses RLS extensions with vector measurements, the sufficient conditions for stability may be overly restrictive when applied to RLS extensions as a much more general class of recursive estimators is analyzed. 

\subsubsection{Robustness}

Several RLS extensions further analyze robustness to time-varying parameters and to bounded measurement noise \cite{Dasgupta1987Asymptotically,samson1983stability,milek1995time}. 
If there is noise in the measurement and regressor, this is known as the \textit{errors-in-variables} problem, which is significantly more challenging than the problem of robustness to measurement noise alone \cite{soderstrom2007errors,griliches1986errors,gillard2010overview}.
In particular, if the measurement noise and regressor noise are correlated, then, with the exception of some special cases, the least squares estimator is asymptotically biased \cite[p. 205]{ljung1998system} and it is often desirable to bound such bias.
Similarly to stability, while \cite{bin2022generalized} addresses the errors-in-variables problem, the required sufficient conditions may be overly restrictive and the bounds on asymptotic bias overly loose when applied specifically to RLS extensions.

\subsection{Contributions}

The contributions of this article are summarized as follows:
\begin{enumerate}[leftmargin=*]
    \item We derive generalized forgetting recursive least squares (GF-RLS), which is a discrete-time version of the continuous-time RLS generalization developed in \cite{shaferman2021continuous}. See section \ref{sec: GFRLS}.
    While \cite{shaferman2021continuous} presents continuous-time versions of RLS algorithms that fit into their framework,
    we will show in section \ref{sec: RLS extensions} how various RLS algorithms from the literature, without any modifications, can be derived from the GF-RLS cost function as special cases. 
    Hence, our subsequent analysis directly applies the original discrete-time RLS algorithms which have been widely used in their original discrete-time form.
    \item For constant-parameter estimation, we use Lyapunov methods to develop stability guarantees for GF-RLS.
    These guarantees extend results obtained in \cite{parkum1992recursive} by generalizing to vector measurements and by providing weaker stability guarantees when not all the conditions for exponential stability are met. See section \ref{sec: GFRLS stability}.
    The sufficient conditions for stability we present are similar to the heuristic conditions in \cite{Salgado1988modified} for design of RLS algorithms.
    Thus, we believe our sufficient conditions may serve as a more precise guideline in the design of future RLS algorithms.
    \item In addition, we develop robustness guarantees of GF-RLS to bounded parameter variation, bounded measurement noise, and bounded regressor noise. In particular, we obtain sufficient conditions for the global uniform ultimate boundedness of the parameter-estimation error.
    None of these three additional factors were considered in \cite{parkum1992recursive} or \cite{shaferman2021continuous}.
    A specialization of this result provides a bound on the asymptotic bias of parameter estimation error in the context of the errors-in-variables problem.
    See section \ref{subsec: GFRLS robustness}.

\end{enumerate}

\subsection{Notation and Terminology}
$\BBN_0$ denotes the set of non-negative integers $\{0,1,2,\hdots\}$.
$I_n$ denotes the $n \times n$ identity matrix, and $0_{m \times n}$ denotes the $m \times n$ zero matrix. 
For symmetric $A\in \BBR^{n \times n}$, let the $n$ real eigenvalues of $A$ be denoted by
$\eigmin(A) \triangleq \boldsymbol{\lambda_n}(A) \le \cdots \le \eigmax(A) \triangleq \boldsymbol{\lambda_1}(A)$. 
For $B \in \BBR^{m \times n}$, $\svdmax(B)$ denotes the largest singular value of $B$, and $\svdmin(B)$ denotes the smallest singular value of $B$.

For symmetric $P,Q \in \BBR^{n \times n}$, $P \prec Q$ (respectively, $P \preceq Q$) denotes that $Q-P$ is positive definite (respectively, positive semidefinite).
For all $x \in \BBR^n$, $\Vert x \Vert$ denotes the Euclidean norm, that is $\Vert x \Vert \triangleq \sqrt{x^\rmT x}$. 
For positive-semidefinite $R \in \BBR^{n \times n}$ and $x \in \BBR^n$, $\Vert x \Vert_R \triangleq \sqrt{x^\rmT R x}$. 
For symmetric $S \in \BBR^{n \times n}$, $\Vert x \Vert_S^2 \triangleq x^\rmT S x$. Note that the notation $\Vert x \Vert_S^2$ is used only for convenience and that $\Vert x \Vert_S$ is not defined when $S$ is not positive semidefinite.
For $\varepsilon > 0$ and $x_e \in \BBR^n$, define the closed ball $\bar{\mathcal{B}}_\varepsilon(x_e) \triangleq \{x \in \BBR^n \colon \Vert x - x_e \Vert \le \varepsilon \}$.

\begin{defin}
    \label{defn: PE}
A sequence $(\phi_k)_{k=k_0}^\infty \subset \BBR^{p \times n}$ is \textit{persistently exciting} if there exist $N \ge 1$ and $\alpha>0$ such that, for all $k \ge k_0$, 
\begin{align}
    \alpha I_n \preceq \sum_{i=k}^{k+N-1}\phi_i^\rmT \phi_i.
\end{align}
Furthermore, $\alpha$ and $N$ are, respectively, the {\it lower bound}
 and {\it persistency window} of $(\phi_k)_{k=k_0}^\infty$.
\end{defin}
\begin{defin}
    \label{defn: BE}
A sequence $(\phi_k)_{k=k_0}^\infty \subset \BBR^{p \times n}$ is \textit{bounded} if there exists
$\beta\in(0,\infty)$
such that, for all $k \ge k_0$, 
\begin{align}
    \label{eqn: bounded sequence}
    \phi_k^\rmT \phi_k \preceq \beta I_n.
\end{align}
Furthermore, $\beta$ is the \textit{upper bound} of $(\phi_k)_{k=k_0}^\infty$.
\end{defin}


\section{Generalized Forgetting Recursive Least Squares (GF-RLS)}
\label{sec: GFRLS}
The following theorem presents generalized forgetting recursive least squares, which is a discrete-time RLS generalization derived from minimizing a least-squares cost function. 

\begin{theo}
\label{theo: GFRLS}
For all $k \ge 0$, let $\Gamma_k \in \BBR^{p \times p}$ be positive definite,
let $\phi_k \in \BBR^{p \times n}$, and let $y_k \in \BBR^{p}$. 
Furthermore, let $P_0 \in \BBR^{n \times n}$ be positive definite, and let $\theta_0 \in \BBR^{n}$. 
For all $k \ge 0$, let $F_k \in \BBR^{n \times n}$ be symmetric and satisfy
\begin{align}
    \label{eqn: Pinv - F is pos def GFRLS}
    F_k \prec P_0^{-1} + \sum_{i=0}^{k-1} \left( - F_i + \phi_i^\rmT \Gamma_i^{-1} \phi_i  \right).
\end{align}
For all $k \ge 0$, define $J_k \colon \BBR^n \rightarrow \BBR$ by
\begin{equation}
    \begin{aligned}
        \label{eqn: GFRLS Cost}
        J_{k}(\hat{\theta}) \triangleq J_{k,{\rm loss}}(\hat{\theta}) - J_{k,{\rm forget}}(\hat{\theta}) + J_{k,{\rm reg}}(\hat{\theta}),
    \end{aligned}
\end{equation}
where
\begin{align}
    J_{k,{\rm loss}}(\hat{\theta}) &\triangleq \sum_{i=0}^{k} \Vert y_i - \phi_i \hat{\theta} \Vert _{\Gamma_i^{-1}}^2 , \\ 
    J_{k,{\rm forget}}(\hat{\theta}) &\triangleq \sum_{i=0}^{k} \Vert \hat{\theta} - \theta_i \Vert_{ F_i}^2 , \\
    J_{k,{\rm reg}}(\hat{\theta}) &\triangleq  \Vert \hat{\theta} - \theta_0 \Vert_{P_0^{-1}}^2.
\end{align}
Then, $J_k$ has a unique global minimizer, denoted
\begin{align}
    \theta_{k+1} \triangleq \argmin_{\hat{\theta} \in \BBR^n} J_k(\hat{\theta}),
\end{align}
which, for all $k \ge 0$, is given by
\begin{align}
    P_{k+1}^{-1} &= P_k^{-1} - F_k  + \phi_k^\rmT \Gamma_k^{-1} \phi_k, \label{eqn: GFRLS Pinv Update}\\
    \theta_{k+1} &= \theta_k + P_{k+1} \phi_k^\rmT \Gamma_k^{-1} (y_k - \phi_k \theta_k ). \label{eqn: GFRLS theta Update}
\end{align}
\end{theo}

\begin{proof}
    See Appendix C.
\end{proof}

For all $k \ge 0$, We call 
$y_k \in \BBR^p$ the \textit{measurement}, 
$\phi_k \in \BBR^{p \times n}$ the \textit{regressor}, 
and $\theta_k \in \BBR^n$ the \textit{parameter estimate}. 
Moreover, we call 
$\Gamma_k \in \BBR^{p \times p}$ the \textit{weighting matrix}, 
$F_k \in \BBR^{n \times n}$ the \textit{forgetting matrix}, 
and $P_k \in \BBR^{n \times n}$ the \textit{covariance matrix}.
Furthermore, \eqref{eqn: GFRLS Pinv Update} and \eqref{eqn: GFRLS theta Update} are the GF-RLS update equations.

Notice that condition \eqref{eqn: Pinv - F is pos def GFRLS} guarantees that $J_k$ has a unique global minimizer, as shown in the proof of Theorem \ref{theo: GFRLS}. Corollary \ref{cor: Pkinv - Fk is pos def GFRLS v2} gives an important interpretation to \eqref{eqn: Pinv - F is pos def GFRLS}.

\begin{cor}
\label{cor: Pkinv - Fk is pos def GFRLS v2}
    Consider the notation and assumptions of Theorem \ref{theo: GFRLS}. Then, for all $k \ge 0$,
    \begin{align}
    \label{eqn: Pinv - F is pos def GFRLS v2}
        P_k^{-1} = P_0^{-1} + \sum_{i=0}^{k-1} \left( - F_i + \phi_i^\rmT \Gamma_i^{-1} \phi_i \right).
    \end{align}
    Hence, for all $k \ge 0$, \eqref{eqn: Pinv - F is pos def GFRLS} holds if and only if
    \begin{align}
    \label{eqn: Pinv - F is pos def explicit}
        P_k^{-1} - F_k \succ 0.
    \end{align}
\end{cor}
\begin{proof}
    Equation \eqref{eqn: Pinv - F is pos def GFRLS v2} follows directly from repeated substitution of \eqref{eqn: GFRLS Pinv Update}. Next, \eqref{eqn: Pinv - F is pos def explicit} follows from substituting \eqref{eqn: Pinv - F is pos def GFRLS v2} into \eqref{eqn: Pinv - F is pos def GFRLS}.
\end{proof}
Corollary \ref{cor: Pkinv - Fk is pos def GFRLS v2} shows that to ensure that the GF-RLS cost \eqref{eqn: GFRLS Cost} has a unique global minimizer (i.e. \eqref{eqn: Pinv - F is pos def GFRLS} is satisfied), it suffices to, for all $k \ge 0$, choose $F_k$ such that $P_k^{-1} - F_k \succ 0$.

\begin{defin}
    \label{defn: GFRLS proper}
    GF-RLS is {\it proper} if, for all $k \ge 0$, $F_k \in \BBR^{n \times n}$ is positive semidefinite. GF-RLS is \textit{improper} if it is not proper.
\end{defin}

Note that the GF-RLS cost $J_k(\hat{\theta})$ is composed as the sum of three terms, namely, 
the {\it loss term} $J_{k,{\rm loss}}(\hat{\theta})$, 
the {\it forgetting term} $-J_{k,{\rm forget}}(\hat{\theta})$, 
and the {\it regularization term} $J_{k,{\rm reg}}(\hat{\theta})$. 
Note that, if GF-RLS is proper, then, for all $\hat{\theta} \in \BBR^n$, the forgetting term $-J_{k,{\rm forget}}(\hat{\theta})$ is nonpositive. 
In practice, if GF-RLS is proper, then the forgetting term rewards the difference between the estimate $\theta_{k+1}$ and $\theta_i$ for previous steps $0 \le i \le k$. 
This reward is weighted by the forgetting matrix $F_k \in \BBR^{n \times n}$. 
It is shown in Section \ref{sec: RLS extensions} that for particular choices of the forgetting matrix, we recover extensions of RLS with forgetting from GF-RLS.

\section{Stability of Fixed Parameter Estimation}
\label{sec: GFRLS stability}

For the analysis of this section, we make the assumption that there exist fixed parameters $\theta \in \BBR^n$ such that, for all $k \ge 0$, 
\begin{align}
    \label{eqn: y = phi theta}
    y_k = \phi_k \theta.
\end{align}
Furthermore, for all $k \ge 0$, we define the parameter estimation error $\tilde{\theta}_k \in \BBR^n$ by
\begin{align}
    \label{eqn: thetatilde defn}
    \tilde{\theta}_k \triangleq \theta_k - \theta.
\end{align}
Substituting into \eqref{eqn: GFRLS theta Update}, it then follows that
%
%
\begin{align}
\label{eqn: theta tilde update}
    \tilde{\theta}_{k+1} = M_k \tilde{\theta}_k,
\end{align}
where, for all $k \ge 0$, $M_k \in \BBR^{n \times n}$ is defined
\begin{align}
    \label{eqn: M_k defn}
    M_k \triangleq I_n - P_{k+1} \phi_k^\rmT \Gamma_k^{-1} \phi_k.
\end{align}
Hence, \eqref{eqn: theta tilde update} is a linear time-varying system with an equilibrium $\tilde{\theta}_k \equiv 0$.

Next, for all $k \ge 0$, let $\Gamma_k^{-\frac{1}{2}} \in \BBR^{p \times p}$ be the unique positive-definite matrix such that
\begin{align}
    \Gamma_k^{-1} = \Gamma_k^{-\frac{1}{2} \rmT} \Gamma_k^{-\frac{1}{2}}.
\end{align}
Furthermore, define the \textit{weighted regressor} $\bar{\phi}_k \in \BBR^{p \times n}$ by
\begin{align}
    \label{eqn: phibar defn}
    \bar{\phi}_k \triangleq \Gamma_k^{-\frac{1}{2}} \phi_k.
\end{align}
Substituting \eqref{eqn: phibar defn} into \eqref{eqn: M_k defn}, it follows that, for all $k \ge 0$,
\begin{align}
    \label{eqn: M_k v2}
    M_k = I_n - P_{k+1} \bar{\phi}_k^\rmT \bar{\phi}_k.
\end{align}

Finally, let $k_0 \ge 0$ and consider the following conditions:
\begin{enumerate}
    \myitem{\textit{A1)}}\label{item: cond F PSD} For all $k \ge k_0$, $F_k \succeq 0$.
    \myitem{\textit{A2)}}\label{item: cond inv(Pkinv - F) upper bound}   There exists $b \in (0,\infty)$ such that, for all $k \ge k_0$, $(P_k^{-1} - F_k)^{-1} \preceq b I_n$. 
    \myitem{\textit{A3)}}\label{item: cond Pk lower bound} There exists $a > 0$ such for all $k \ge k_0$, $a I_n \preceq P_k$. 
    \myitem{\textit{A4)}}\label{item: cond weighted regressor PE and bounded} The sequence of weighted regressors $(\bar{\phi}_k)_{k=k_0}^\infty$ is persistently exciting with lower bound $\bar{\alpha} > 0$ and persistency window $N \ge 1$ and bounded with upper bound $\bar{\beta} \in (0,\infty)$.
\end{enumerate}
We now present Theorem \ref{theo: GFRLS Lyapunov stability v2} which gives sufficient conditions for the stability of the equilibrium $\tilde{\theta}_k \equiv 0$ of \eqref{eqn: theta tilde update}.
See Appendix B for relevant discrete-time stability definitions.

\begin{theo}
\label{theo: GFRLS Lyapunov stability v2}
For all $k \ge 0$, let $\Gamma_k \in \BBR^{p \times p}$ be positive definite,
let $\phi_k \in \BBR^{p \times n}$, let $y_k \in \BBR^{p}$, and let $F_k \in \BBR^{n \times n}$ be symmetric and satisfy \eqref{eqn: Pinv - F is pos def GFRLS}.
Let $P_0 \in \BBR^{n \times n}$ be positive definite, and let $\theta_0 \in \BBR^{n}$. 
For all $k \ge 1$, let $P_k \in \BBR^{n \times n}$ and $\theta_k \in \BBR^n$ be recursively updated by \eqref{eqn: GFRLS Pinv Update} and \eqref{eqn: GFRLS theta Update}.
%
%
Furthermore, assume there exists $\theta \in \BBR^n$ such that, for all $k \ge 0$, \eqref{eqn: y = phi theta} holds.
Then the following statements hold:
    \begin{enumerate}
        \myitem{1)}\label{item: stability statement 1} If there exists $k_0 \ge 0$ such that \ref{item: cond F PSD} and \ref{item: cond inv(Pkinv - F) upper bound}, then the equilibrium $\tilde{\theta}_k \equiv 0$ of \eqref{eqn: theta tilde update} is Lyapunov stable.
        \myitem{2)}\label{item: stability statement 2} If there exists $k_0 \ge 0$ such that \ref{item: cond F PSD}, \ref{item: cond inv(Pkinv - F) upper bound}, and \ref{item: cond Pk lower bound}, then the equilibrium $\tilde{\theta}_k \equiv 0$ of \eqref{eqn: theta tilde update} is uniformly Lyapunov stable.
        \myitem{3)}\label{item: stability statement 3} If there exists $k_0 \ge 0$ such that \ref{item: cond F PSD}, \ref{item: cond inv(Pkinv - F) upper bound}, and \ref{item: cond weighted regressor PE and bounded}, then the equilibrium $\tilde{\theta}_k \equiv 0$ of \eqref{eqn: theta tilde update} is globally asymptotically stable.\footnote{\label{note1}Since \eqref{eqn: theta tilde update} is linear time-varying (LTV), it suffices to say that \eqref{eqn: theta tilde update} is asymptotically (resp. uniformly exponentially) stable, as asymptotic (resp. uniform exponential) stability implies global asymptotic (resp. uniform exponential) stability for LTV systems \cite{zhou2017asymptotic}, \cite[sec. 5.5]{agarwal2000difference}.}
        \myitem{4)}\label{item: stability statement 4} If there exists $k_0 \ge 0$ such that \ref{item: cond F PSD}, \ref{item: cond inv(Pkinv - F) upper bound}, \ref{item: cond Pk lower bound}, and \ref{item: cond weighted regressor PE and bounded}, then the equilibrium $\tilde{\theta}_k \equiv 0$ of \eqref{eqn: theta tilde update} is globally uniformly exponentially stable.\footnoteref{note1}
    \end{enumerate}
\end{theo}

\begin{proof}
     In the case $k_0 = 0$, see Appendix D for a proof of statements \ref{item: stability statement 1} and \ref{item: stability statement 2}, and Appendix E for a proof of statements \ref{item: stability statement 3} and \ref{item: stability statement 4}. The case $k_0 \ge 1$ can be shown similarly.
     Also see Figure \ref{fig:proof-roadmap} for a proof roadmap of Theorem 2.
\end{proof}

\subsection{Discussion of Conditions \ref{item: cond F PSD} through \ref{item: cond weighted regressor PE and bounded}}
\label{subsec: Theo conditions 1 - 4 attainability}

This subsection gives a brief discussion of conditions \ref{item: cond F PSD} through \ref{item: cond weighted regressor PE and bounded} used in Theorem \ref{theo: GFRLS Lyapunov stability v2}.

\subsubsection{Condition \ref{item: cond F PSD}} 

Note that by Definition \ref{defn: GFRLS proper}, this condition is equivalent to GF-RLS being proper.
Furthermore, whether or not GF-RLS is proper is a direct consequence of the algorithm design.
We will show in section \ref{sec: RLS extensions} how ten different extensions of RLS are all proper (some requiring minor assumptions), and hence satisfy condition \ref{item: cond F PSD}. 

\subsubsection{Conditions \ref{item: cond inv(Pkinv - F) upper bound} and \ref{item: cond Pk lower bound}} 

In 1988, \cite{Salgado1988modified} qualitatively proposed that RLS extensions should guarantee a (nonzero) lower bound and a (noninfinite) upper bound of the covariance matrix $P_k$ for good performance. That is, $a > 0$ and $b \in (0,\infty)$ such that, for all $k \ge 0$, $a I_n \preceq P_k \preceq b I_n$.
Many RLS extensions since have provided analysis which guarantee an upper and lower bound of the covariance matrix \cite{cao2000directional,Lai2023Exponential,goel2020recursive,bittanti1990convergence}.
A lower bound on $P_k$ is equivalent to condition \ref{item: cond Pk lower bound}. 
However, an upper bound on $P_k$ does not guarantee condition \ref{item: cond inv(Pkinv - F) upper bound}.
%
%
%

Nevertheless, for many choices of $F_k$ from the literature, condition \ref{item: cond inv(Pkinv - F) upper bound} follows easily from an upper bound on the covariance matrix.
%
%
A future area of interest is whether similar stability and robustness guarantees exist if condition \ref{item: cond inv(Pkinv - F) upper bound} is replaced with an upper bound on $P_k$.

\subsubsection{Condition \ref{item: cond weighted regressor PE and bounded}} 

Persistent excitation and boundedness of the sequence of regressors $(\phi_k )_{k=k_0}^\infty$ is an important requirement for convergence in RLS extensions \cite{johnstone1982exponential,Salgado1988modified}. 
While work has been done the relax the persistent excitation condition \cite{efimov2018robust,wang2020fixed,vamvoudakis2015asymptotically}, it has been shown that \textit{weak persistent excitation} is necessary for the global asymptotic stability of RLS \cite{bruce2021necessary}.

Note that condition \ref{item: cond weighted regressor PE and bounded} requires persistent excitation and boundedness of the the sequence of weighted regressors $( \bar{\phi}_k )_{k = k_0}^\infty$, rather than the sequence of regressors $(\phi_k )_{k=k_0}^\infty$.
Corollary \ref{cor: weighted regressors} gives a sufficient condition for when persistent excitation and boundedness of the sequence of regressors $(\phi_k)_{k=k_0}^\infty$ implies persistent excitation and boundedness of the sequence of weighted regressors $(\bar{\phi}_k)_{k=k_0}^\infty$.
Note, however, that the bounds guaranteed by Corollary \ref{cor: weighted regressors} are often loose and it is preferable in practice to directly analyze the sequence of weighted regressors.

\begin{cor}
    \label{cor: weighted regressors}
    Assume there exist $k_0 \ge 0$ and $0 < \gamma_{\rm min} < \gamma_{\rm max}$ such that, for all $k \ge k_0$, 
    \begin{align}
        \label{eqn: Gamma_k bounds}
        \gamma_{\rm min} I_p \preceq \Gamma_k \preceq \gamma_{\rm max} I_p.
    \end{align}
    Furthermore, let $(\phi_k)_{k=k_0}^\infty$ be persistently exciting with lower bound $\alpha > 0$ and persistency window $N$ and bounded with upper bound $\beta \in (0,\infty)$.
    Then, $(\bar{\phi}_k)_{k=k_0}^\infty$ is persistently exciting with lower bound $\frac{\alpha}{\gamma_{\rm max}}$ and persistency window $N$ and bounded with upper bound $\frac{\beta}{\gamma_{\rm min}}$. 
\end{cor}

\begin{proof}
    Note that, for all $k \ge k_0$, $\bar{\phi}_k^\rmT \bar{\phi}_k = \phi_k^\rmT \Gamma_k^{-1} \phi_k \succeq \frac{1}{\gamma_{\rm max}} \phi_k^\rmT \phi_k$, and hence $\sum_{i=k}^{k+N-1}\bar{\phi}_i^\rmT \bar{\phi}_i \succeq \sum_{i=k}^{k+N-1} \frac{1}{\gamma_{\rm max}} \phi_i^\rmT \phi_i \succeq \frac{\alpha}{\gamma_{\rm max}} I_n$.
\end{proof}

\section{Robustness to Time-Varying Parameters, Measurement Noise, and Regressor Noise}
\label{subsec: GFRLS robustness}
For the analysis of this section, we make the assumption that, for all $k \ge 0$, the parameters $\theta_{{\rm true},k} \in \BBR^{n}$ are time-varying and satisfy
\begin{align}
\label{eqn: theta_true update combo}
    \theta_{{\rm true},k+1} = \theta_{{\rm true},k} + \delta_{\theta,k},
\end{align}
where, for all $k \ge 0$, $\delta_{\theta,k} \in \BBR^n$ is the \textit{change in the parameters}. Note that no model of how the parameters evolve is known. 
Furthermore, assume that, for all $k \ge 0$, 
\begin{align}
    \label{eqn: y_k defn time varying parameters with noise}
    y_k = (\phi_k + \delta_{\phi,k}) \theta_{{\rm true},k} + \delta_{y,k},
\end{align}
where $\delta_{y,k} \in \BBR^p$ is the \textit{measurement noise} and $\delta_{\phi,k} \in \BBR^{p \times n}$ is the \textit{regressor noise}. Furthermore, for all $k \ge 0$, we define the \textit{weighted measurement noise} at step $k$, $\bar{\delta}_{y,k} \in \BBR^p$, and the \textit{weighted regressor noise}, $\bar{\delta}_{\phi_k} \in \BBR^{p \times n}$, by
\begin{align}
    \bar{\delta}_{y,k} &\triangleq \Gamma_k^{-\frac{1}{2}} \delta_{y,k}, \label{eqn: weighted measurement noise} \\
    \bar{\delta}_{\phi,k} &\triangleq \Gamma_k^{-\frac{1}{2}} \delta_{\phi,k}. \label{eqn: weighted regressor noise}
\end{align}

Note that, for all $k \ge 1$, the parameter estimate $\theta_k$ is based on measurements up to step $k-1$, that is, $\{ y_0,\hdots,y_{k-1} \}$. To compensate for this one-step delay, we define, for all $k \ge 0$, the parameter estimation error $\check{\theta}_k \in \BBR^n$ by
\begin{align}
\label{eqn: theta bar defn}
    \check{\theta}_k \triangleq \theta_k - \theta_{{\rm true},k-1}.
\end{align}
Substituting \eqref{eqn: y_k defn time varying parameters with noise} and \eqref{eqn: theta bar defn} into \eqref{eqn: GFRLS theta Update} implies that, for all $k \ge 0$,
\begin{align*}
    \check{\theta}_{k+1} = M_k (\check{\theta}_k - \delta_{\theta,k-1}) + P_{k+1} \phi_k^\rmT \Gamma_k^{-1} (\delta_{\phi,k} \theta_{{\rm true},k} + \delta_{y,k}),
\end{align*}
and then substituting \eqref{eqn: phibar defn}, \eqref{eqn: weighted measurement noise}, and \eqref{eqn: weighted regressor noise}, it follows that
\begin{align}
\label{eqn: theta bar update}
    \check{\theta}_{k+1} = M_k (\check{\theta}_k - \delta_{\theta,k-1}) + P_{k+1} \bar{\phi}_k^\rmT (\bar{\delta}_{\phi,k} \theta_{{\rm true},k} + \bar{\delta}_{y,k}),
\end{align}
Hence, \eqref{eqn: theta bar update} is a nonlinear system $\check{\theta}_{k+1} = \check{f}(k,\check{\theta}_{k})$. 

Finally, let $k_0 \ge 0$ and consider the following conditions:
\begin{enumerate}
    \myitem{\textit{A5)}}\label{item: cond delta theta bound} There exists $\delta_{\theta} \ge 0$ such that, for all $k \ge k_0$, $\Vert \delta_{\theta,k} \Vert \le \delta_{\theta}.$
    \myitem{\textit{A6)}}\label{item: cond delta y bound} There exists $\bar{\delta}_{y} \ge 0$ such that, for all $k \ge k_0$, $\Vert \bar{\delta}_{y,k} \Vert \le \bar{\delta}_{y}.$
    \myitem{\textit{A7)}}\label{item: cond delta phi bound} There exists $\bar{\delta}_{\phi} \ge 0$ such that the sequence $(\bar{\delta}_{\phi,k})_{k=k_0}^\infty$ is bounded with upper bound $\bar{\delta}_{\phi}$.
    \myitem{\textit{A8)}}\label{item: cond theta true bound} There exists $\theta_{\rm max} \ge 0$ such that, for all $k \ge k_0$, $\Vert \theta_{{\rm true},k} \Vert \le \theta_{\rm max}.$
\end{enumerate}
We now present Theorem \ref{theo: GFRLS UUB} which gives sufficient conditions for the global uniform ultimate boundedness of \eqref{eqn: theta bar update}.
Please see Appendix B for the definition of global uniform ultimate boundedness.

\begin{theo}
\label{theo: GFRLS UUB}
For all $k \ge 0$, let $\Gamma_k \in \BBR^{p \times p}$ be positive definite,
let $\phi_k \in \BBR^{p \times n}$, let $y_k \in \BBR^{p}$, and let $F_k \in \BBR^{n \times n}$ be symmetric and satisfy \eqref{eqn: Pinv - F is pos def GFRLS}.
Let $P_0 \in \BBR^{n \times n}$ be positive definite, and let $\theta_0 \in \BBR^{n}$. 
For all $k \ge 1$, let $P_k \in \BBR^{n \times n}$ and $\theta_k \in \BBR^n$ be recursively updated by \eqref{eqn: GFRLS Pinv Update} and \eqref{eqn: GFRLS theta Update}.
%
%
Furthermore, for all $k \ge 0$, let $\theta_{{\rm true},k} \in \BBR^n$, $\delta_{\theta,k} \in \BBR^n$, $\delta_{y,k} \in \BBR^p$, and $\delta_{\phi,k} \in \BBR^{p \times n}$ satisfy \eqref{eqn: theta_true update combo} and \eqref{eqn: y_k defn time varying parameters with noise}.
Finally, let $k_0 \ge 0$ be such that conditions \ref{item: cond F PSD}, \ref{item: cond inv(Pkinv - F) upper bound}, \ref{item: cond Pk lower bound}, \ref{item: cond weighted regressor PE and bounded}, \ref{item: cond delta theta bound}, \ref{item: cond delta y bound}, \ref{item: cond delta phi bound}, and \ref{item: cond theta true bound} hold.
Then, the system \eqref{eqn: theta bar update} is globally uniformly ultimately bounded with bound $\varepsilon$ given by 
\begin{align}
\label{eqn: epsilonbar UUB time-varying}
    \varepsilon &= \varepsilon^* \left[ \delta_\theta +  b \bar{\beta}^{\frac{1}{2}} 
    \left( \bar{\delta}_{\phi}^{\frac{1}{2}} \theta_{\rm max} + \bar{\delta}_y \right) \right],
\end{align}
where
%
\begin{gather}
    \label{eqn: epsilonbar star defn}
    \varepsilon^* \triangleq \max \left\{ 1, \frac{1}{\sqrt{a}} \right\}
    \left(\Delta_N + \sqrt{\Delta_N + \Delta_N^2} \right) N,
    \\
    \label{eqn: Delta_N bar defn}
    \Delta_N \triangleq \frac{N}{a \bar{\alpha}} 
    \left(1 + b \bar{\beta} \right)
    \left[ 1 + \frac{N-1}{2} \left(b \bar{\beta}\right)^2 \right] - 1.
\end{gather}
\end{theo}
\begin{proof}
    We prove the case $k_0 = 0$. The case $k_0 \ge 1$ can be shown similarly.
    Note that, for all $k \ge 0$, \eqref{eqn: theta bar update} can be written as 
    \begin{align}
    \label{eqn: thetacheck update expanded}
        \check{\theta}_{k+1} = M_k( \check{\theta}_k - \delta_{\theta,k-1} + M_k^{-1} P_{k+1}  \bar{\phi}_k^\rmT (\bar{\delta}_{\phi,k} \theta_{{\rm true},k} + \bar{\delta}_{y,k})).
    \end{align}
    Moreover, it follows from \eqref{eqn: GFRLS Pinv Update} and \eqref{eqn: M_k defn} that, for all $k \ge 0$, $M_k = P_{k+1}(P_k^{-1} - F_k)$. It follows from Corollary \ref{cor: Pkinv - Fk is pos def GFRLS v2} that, for all $k \ge 0$, $(P_k^{-1} - F_k)$ is nonsingular, and hence 
    \begin{align}
    \label{eqn: M_k inv identity prep}
        M_k^{-1} = (P_k^{-1} - F_k)^{-1} P_{k+1}^{-1}.
    \end{align}
    Substituting \eqref{eqn: M_k inv identity prep} into \eqref{eqn: thetacheck update expanded} then gives, for all $k \ge 0$,  
    \begin{align*}
        \check{\theta}_{k+1} = M_k( \check{\theta}_k - \zeta_k),
    \end{align*}
    where $\zeta_k \in \BBR^n$ is defined
    \begin{align}
    \label{zeta_k defn}
        \zeta_k \triangleq \delta_{\theta,k-1} - (P_k^{-1} - F_k)^{-1} \bar{\phi}_k^\rmT (\bar{\delta}_{\phi,k} \theta_{{\rm true},k} +  \bar{\delta}_{y,k}).
    \end{align}

    Next, it follows from applying triangle inequality and norm sub-multiplicativity to \eqref{zeta_k defn} and using the bounds in conditions \ref{item: cond inv(Pkinv - F) upper bound}, \ref{item: cond weighted regressor PE and bounded}, \ref{item: cond delta theta bound}, \ref{item: cond delta y bound}, \ref{item: cond delta phi bound}, and \ref{item: cond theta true bound} that, for all $k \ge 0$,
    \begin{align*}
        \Vert \zeta_k \Vert \le \delta_\theta +  b \bar{\beta}^{\frac{1}{2}} 
        \left( \bar{\delta}_{\phi}^{\frac{1}{2}} \theta_{\rm max} + \bar{\delta}_y \right) \triangleq \zeta.
    \end{align*}
    Finally, it follows Lemma \ref{lem: GFRLS UUB} in Appendix F that the system \eqref{eqn: theta bar update} is globally uniformly ultimately bounded with bound $\varepsilon^* \zeta$.
    For further details, see Figure \ref{fig:proof-roadmap} for a proof roadmap of Theorem 3.
\end{proof}

As a sanity check, note that, for all $N \ge 1$, $\Delta_N > \frac{Nb \bar{\beta}}{a \bar{\alpha}} - 1 \ge N-1 \ge 0$, and hence $\varepsilon^* > 0$. Therefore, if $\delta_\theta > 0$, $\bar{\delta}_y > 0$, or $\bar{\delta}_{\phi}^{\frac{1}{2}} \theta_{\rm max} > 0$, then $\varepsilon > 0$.

\subsection{Discussion of Conditions \ref{item: cond delta theta bound} through \ref{item: cond theta true bound}}

This subsection gives a discussion of conditions \ref{item: cond delta theta bound} through \ref{item: cond theta true bound} used in Theorem \ref{theo: GFRLS UUB}. 
See subsection \ref{subsec: Theo conditions 1 - 4 attainability} for a discussion of conditions \ref{item: cond F PSD} through \ref{item: cond weighted regressor PE and bounded}.

\subsubsection{Condition \ref{item: cond delta theta bound}}

Condition \ref{item: cond delta theta bound} is a bound on how quickly the parameters being estimated, $\theta_{{\rm true},k}$, can change. While these parameters are not known, in practice, this bound can be estimated from data.

\subsubsection{Conditions \ref{item: cond delta y bound} and \ref{item: cond delta phi bound}}

Conditions \ref{item: cond delta y bound} and \ref{item: cond delta phi bound} are, respectively, bounds on the weighted measurement noise and weighted regressor noise.
While noise from certain distributions has no guaranteed bound (e.g. Gaussian noise), in practice these bounds can be approximated from data.

Corollary \ref{cor: weighted noise} gives a sufficient condition for when bounded measurement noise and bounded regressor noise imply, respectively, bounded weighted measurement noise and bounded weighted regressor noise.
Note, however, that the bounds guaranteed by Corollary \ref{cor: weighted regressors} are often loose and it is preferable in practice to directly analyze the weighted measurement noise and weighted regressor noise.

\begin{cor}
    \label{cor: weighted noise}
    Assume there exist $k_0 \ge 0$ and $0 < \gamma_{\rm min} < \gamma_{\rm max}$ such that, for all $k \ge k_0$, \eqref{eqn: Gamma_k bounds} holds. Then, the following statements hold: 
    \begin{enumerate}[leftmargin=*]
         \myitem{1)}\label{item: cor statement 2} If there exists $\delta_y \ge 0$ such that, for all $k \ge k_0$, $\Vert \delta_{y,k} \Vert \le \delta_y$, then, for all $k \ge k_0$, $\Vert \bar{\delta}_{y,k} \Vert \le \nicefrac{\delta_y}{\sqrt{\gamma_{\rm min}}}$.
         \myitem{2)}\label{item: cor statement 3} If there exists $\delta_\phi \ge 0$ such that $(\delta_{\phi,k})_{k=k_0}^\infty$ is bounded with upper bound $\delta_{\phi}$, then $(\bar{\delta}_{\phi,k})_{k=k_0}^\infty$ is bounded with upper bound $\frac{\delta_\phi}{\gamma_{\rm min}}$.
    \end{enumerate}
\end{cor}

\begin{proof}
    To show \ref{item: cor statement 2}, note that, for all $k \ge k_0$, $\Vert \Gamma_k^{-\frac{1}{2}} \xi_k \Vert \le \svdmax(\Gamma_k^{-\frac{1}{2}}) \Vert \xi_k \Vert \le \frac{\xi}{\sqrt{\gamma_{\rm min}}}$.
    Lastly, to show \ref{item: cor statement 3}, note that, for all $k \ge k_0$, $\bar{\delta}_{\phi,k}^\rmT \bar{\delta}_{\phi,k} = \delta_{\phi,k}^\rmT \Gamma_k^{-1} \delta_{\phi,k} \preceq \frac{1}{\gamma_{\rm min}} \delta_{\phi,k}^\rmT \delta_{\phi,k} \preceq \frac{\delta_\phi}{\gamma_{\rm min}} I_n$. 
\end{proof}

\subsubsection{Condition \ref{item: cond theta true bound}}
Condition \ref{item: cond theta true bound} is a bound on the magnitude of the parameters being estimated. While the parameters $\theta_{{\rm true},k}$ are not known, this bound can also be approximated in practice.

\subsection{Specialization to Errors-in-Variables}

An important specialization of Theorem \ref{theo: GFRLS UUB} is the case of fixed parameters (i.e. $\delta_\theta = 0$). In this case, only the effect of measurement noise and regressor noise is considered, a problem known as errors-in-variables \cite{soderstrom2007errors}. 
Note that the measurement noise and regressor noise may be correlated, resulting in an asymptotically biased least squares estimator \cite[p. 205]{ljung1998system}.
If parameters are fixed ($\delta_\theta = 0$), it follows from Theorem \ref{theo: GFRLS UUB} that \eqref{eqn: theta bar update} is globally uniformly ultimately bounded with bound 
$\varepsilon =  \varepsilon^*  b \bar{\beta}^{\frac{1}{2}} 
\left( \bar{\delta}_{\phi}^{\frac{1}{2}} \theta_{\rm max} + \bar{\delta}_y \right)$.

More generally, Theorem \ref{theo: GFRLS UUB} can be specialized to assume fixed parameters by setting $\delta_\theta = 0$ and/or to assume no measurement noise by setting $\bar{\delta}_y = 0$ and/or to assume no regressor noise by setting $\bar{\delta}_\phi = 0$. 
As a sanity check, note that if $\delta_\theta = \bar{\delta}_y = \bar{\delta}_\phi = 0$, then \eqref{eqn: epsilonbar UUB time-varying} simplifies to $\varepsilon = 0$.

\section{RLS Extensions as Special Cases of GF-RLS}
\label{sec: RLS extensions}

This section shows how several extensions of recursive least squares with forgetting are special cases of generalized forgetting recursive least squares. For simplicity, we assume that, for all $k \ge 0$, $\Gamma_k = I_p$ in GF-RLS. 
This uniform weighting is in accordance with the RLS extensions we present as originally published. However, these methods can easily be extended to nonuniform weighting by, for all $k \ge 0$, selecting positive-definite $\Gamma_k \in \BBR^{p \times p}$. 
Thereafter, only the forgetting matrix $F_k$ needs to be specified for all $k \ge 0$. Furthermore, the stability results presented in Theorem \ref{theo: GFRLS Lyapunov stability v2} and robustness results presented in Theorem \ref{theo: GFRLS UUB} apply to any algorithm that is a special case of GF-RLS.
For all the following methods, for all $k \ge 0$, let $\phi_k \in \BBR^{p \times n}$ and $y_k \in \BBR^p$. Furthermore let $P_0 \in \BBR^{n \times n}$ be positive definite and $\theta_0 \in \BBR^n$.
If an extension of RLS is a special case of proper GF-RLS, we say that extension is proper. 
Note that we have made minor notational changes to some RLS extensions in order to present all algorithms with the same notation. 
Otherwise, we have done our best to present all algorithms as originally published.
A flowchart summary of this section is given in Figure \ref{fig:GFRLS_flowchart}.

\begin{figure*}
    \centering
    \includegraphics[width = .95\textwidth]{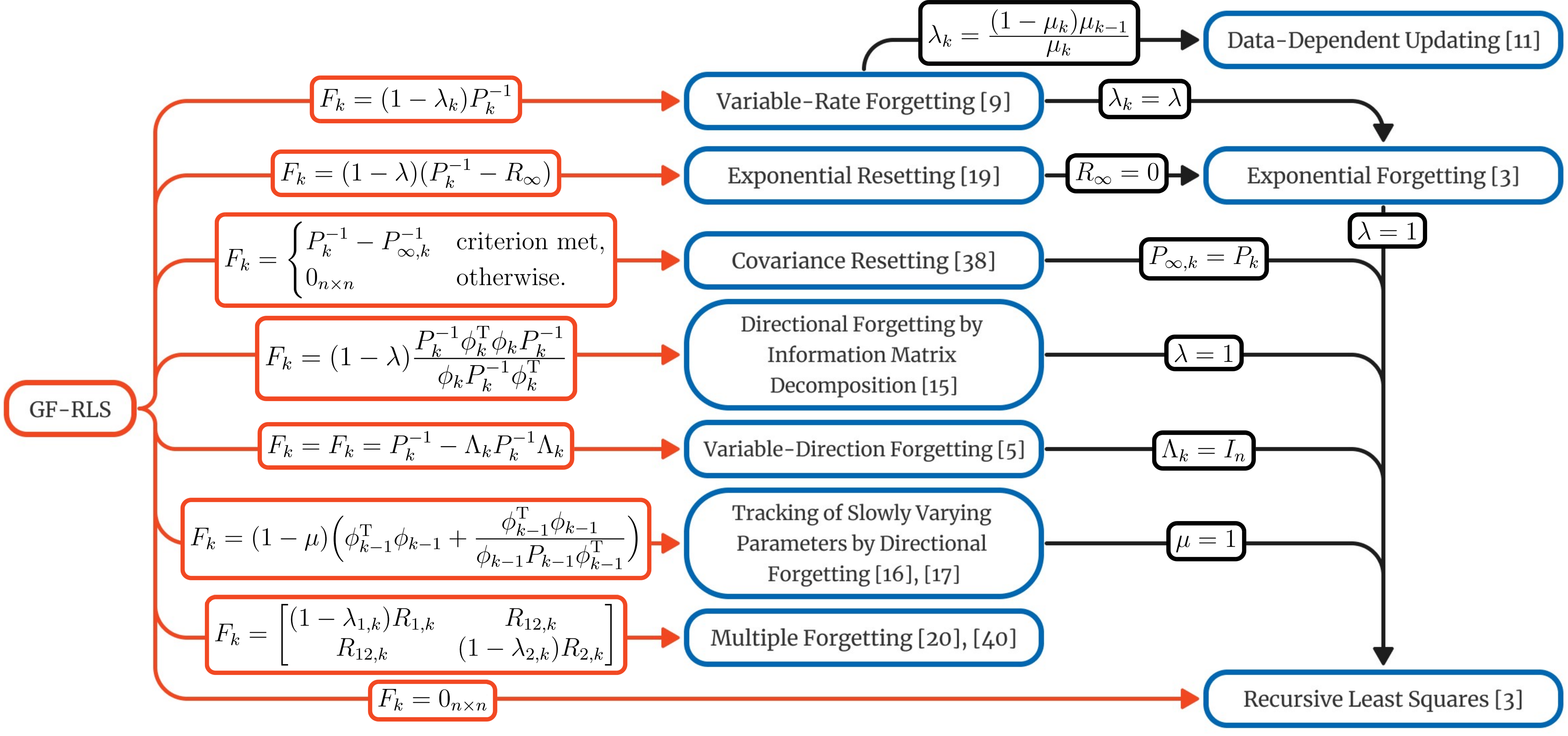}
    \caption{This flowchart summarizes how different extensions of RLS (blue) can be derived as special cases of GF-RLS (red). Furthermore, this chart summarizes how certain RLS extensions are special cases of other RLS extensions (black).}
    \label{fig:GFRLS_flowchart}
\end{figure*}

\subsection{Recursive Least Squares}
\label{subsec: RLS}

Recursive least squares \cite{islam2019recursive} is derived by denoting the minimizer of the cost function
\begin{align}
    J_k(\hat{\theta}) = \sum_{i=0}^k \Vert y_i - \phi_i \hat{\theta} \Vert ^2  +  \Vert \theta - \theta_0 \Vert_{P_0^{-1}}^2
\end{align}
by $\theta_{k+1} \triangleq \argmin_{\hat{\theta} \in \BBR^n} J_k(\hat{\theta}).$
It follows that, for all $k \ge 0$, $\theta_{k+1}$ is given by
\begin{align}
    P_{k+1}^{-1} &= P_{k}^{-1} + \phi_k^\rmT \phi_k, \label{eqn: RLS Pinv update} \\
    \theta_{k+1} &= \theta_k + P_{k+1} \phi_k^\rmT (y_k - \phi_k \theta_k).  \label{eqn: RLS theta update}
\end{align}

Comparing \eqref{eqn: RLS Pinv update} and \eqref{eqn: RLS theta update} to \eqref{eqn: GFRLS Pinv Update} and \eqref{eqn: GFRLS theta Update}, it follows that recursive least squares is a special case of GF-RLS where, for all $k \ge 0$, $\Gamma_k = I_p$ and 
\begin{align}
    F_k = 0_{n \times n}.
\end{align}
Note that, for all $k \ge 0$, $P_k^{-1} \succ 0$, hence $P_k^{-1} - F_k \succ 0$ and $F_k \succeq 0$. Therefore recursive least squares is proper.
%


\subsection{Exponential Forgetting}
A classical method to introduce forgetting in RLS is called \textit{exponential forgetting}, where a forgetting factor $0 < \lambda \le 1$ is introduced which provides exponentially higher weighting to more recent measurements and data \cite{islam2019recursive,goel2020recursive}. Exponential forgetting RLS is derived by denoting the minimizer of the cost function
\begin{align}
    J_k(\hat{\theta}) = \sum_{i=0}^k \lambda^{k-i} \Vert y_i - \phi_i \hat{\theta} \Vert ^2 + \lambda^{k+1} \Vert \theta - \theta_0 \Vert_{P_0^{-1}}^2
\end{align}
by $\theta_{k+1} \triangleq \argmin_{\hat{\theta} \in \BBR^n} J_k(\hat{\theta}).$
It follows that, for all $k \ge 0$, $\theta_{k+1}$ is given by
\begin{align}
    P_{k+1}^{-1} &= \lambda P_{k}^{-1} + \phi_k^\rmT \phi_k, \label{eqn: EF-RLS Pinv update} \\
    \theta_{k+1} &= \theta_k + P_{k+1} \phi_k^\rmT (y_k - \phi_k \theta_k).  \label{eqn: EF-RLS theta update}
\end{align}

Comparing \eqref{eqn: EF-RLS Pinv update} and \eqref{eqn: EF-RLS theta update} to \eqref{eqn: GFRLS Pinv Update} and \eqref{eqn: GFRLS theta Update}, it follows that exponential forgetting is a special case of GF-RLS where, for all $k \ge 0$, $\Gamma_k = I_p$ and 
\begin{align}
    F_k = (1-\lambda)P_k^{-1}.
\end{align}
Note that, for all $k \ge 0$, $P_k^{-1} \succ 0$, hence $P_k^{-1} - F_k = \lambda P_k^{-1} \succ 0$ and $F_k \succeq 0$. Therefore exponential forgetting is proper.


\subsection{Variable-Rate Forgetting }
\label{subsec: VRF}
An extension of exponential forgetting is \textit{variable-rate forgetting}, in which a time-varying forgetting factor, $0 < \lambda_k \le 1$, is selected at each step $k \ge 0$, in place of the constant forgetting factor of exponential forgetting. Variable-rate forgetting is derived in \cite{bruce2020convergence} by defining the cost function
\begin{align}
    J_k(\hat{\theta}) = \sum_{i=0}^k \frac{\rho_k}{\rho_i} \Vert y_i - \phi_i \hat{\theta} \Vert ^2 + \rho_k  \Vert \theta - \theta_0 \Vert_{P_0^{-1}}^2,
\end{align}
where, for all $k \ge 0$, $\rho_k \triangleq \prod_{i=0}^k \lambda_i$. If, for all $k \ge 0$, the minimizer of $J_k(\hat{\theta})$ is denoted by $\theta_{k+1} \triangleq \argmin_{\hat{\theta} \in \BBR^n} J_k(\hat{\theta})$, it follows that, for all $k \ge 0$, $\theta_{k+1}$ is given by
\begin{align}
     P_{k+1}^{-1} &= \lambda_k P_{k}^{-1} + \phi_k^\rmT \phi_k, \label{eqn: VRF-RLS Pinv update} \\
    \theta_{k+1} &= \theta_k + P_{k+1} \phi_k^\rmT (y_k - \phi_k \theta_k). \label{eqn: VRF-RLS theta update}
\end{align}
%

Comparing \eqref{eqn: VRF-RLS Pinv update} and \eqref{eqn: VRF-RLS theta update} to \eqref{eqn: GFRLS Pinv Update} and \eqref{eqn: GFRLS theta Update}, it follows that variable-rate forgetting is a special case of GF-RLS where, for all $k \ge 0$, $\Gamma_k = I_p$ and 
\begin{align}
    F_k = (1-\lambda_k)P_k^{-1}.
\end{align}
Note that, for all $k \ge 0$, $P_k^{-1} \succ 0$, hence $P_k^{-1} - F_k = \lambda_k P_k^{-1} \succ 0$ and $F_k \succeq 0$. Therefore, variable-rate forgetting is proper.

Many methods exist to design this time-varying forgetting factor including methods assuming known noise variance \cite{fortescue1981implementation}, online estimation of noise power \cite{paleologu2008robust}, gradient-based methods \cite{Hung2005Gradient}, and statistical methods \cite{mohseni2022recursive}. 

\subsection{Data-Dependent Updating}
\label{subsec: data dependent updating}

Data-dependent updating was developed in \cite{Dasgupta1987Asymptotically} and was inspired as a way to prevent instabilities in the presence of bounded output disturbances. Data-dependent updating can be summarized by the update equations
\begin{align}
    P_{k+1}^{-1} &= (1-\mu_k) P_k^{-1} + \mu_k \phi_k^\rmT \phi_k, \label{eqn: DDU Pinv update} \\
    \theta_{k+1} &= \theta_k + \mu_k P_{k+1} \phi_k^\rmT (y_k - \phi_k \theta_k), \label{eqn: DDU theta update}
\end{align}
where, for all $k \ge 0$, $0 \le \mu_k < 1$. Next, for all $k \ge 0$, define $\bar{P}_k \in \BBR^{n \times n}$ by
\begin{align}
    \bar{P}_k \triangleq \mu_{k-1} P_k,
\end{align}
where $\mu_{-1} \triangleq 1$. It then follows that, for all $k \ge 0$, \eqref{eqn: DDU Pinv update} and \eqref{eqn: DDU theta update} can be written as
\begin{align}
    \bar{P}_{k+1}^{-1} &= \frac{(1-\mu_k)\mu_{k-1}}{\mu_k} \bar{P}_{k}^{-1} + \phi_k^\rmT \phi_k, \label{eqn: DDU Pinv update 2} \\
    \theta_{k+1} &= \theta_k + \bar{P}_{k+1} \phi_k^\rmT (y_k - \phi_k \theta_k). \label{eqn: DDU theta update 2}
\end{align}
Comparing \eqref{eqn: DDU Pinv update 2} and \eqref{eqn: DDU theta update 2} to \eqref{eqn: VRF-RLS Pinv update} and \eqref{eqn: VRF-RLS theta update}, it follows that data-dependent updating is simply a special case of variable-rate forgetting where, for all $k \ge 0$,
\begin{align}
    \lambda_k = \frac{(1-\mu_k)\mu_{k-1}}{\mu_k}.
\end{align}
For connections to GF-RLS, see subsection \ref{subsec: VRF} on variable-rate forgetting.

 \subsection{Exponential Resetting}
\label{subsec: exponential forgetting}

Exponential resetting was developed in \cite{Lai2023Exponential} 
and can be summarized by the update equations
\begin{align}
    P_{k+1}^{-1} &= \lambda P_k^{-1} + (1-\lambda)R_\infty + \phi_k^\rmT \phi_k, \label{eqn: Pinv update exp reset} \\
    \theta_{k+1} &= \theta_k + P_{k+1} \phi_k^\rmT(y_k - \phi_k \theta_k), \label{eqn: theta update exp reset}
\end{align}
where $R_\infty \in \BBR^{n \times n}$ is positive semidefinite. Note that while \cite{Lai2023Exponential} assumes that $R_\infty$ is positive definite, it is simple to extend the results of \cite{Lai2023Exponential} to positive-semidefinite $R_\infty$.  It is shown in \cite{Lai2023Exponential} that, for all $k \ge 0$, $P_k$ is positive definite. Furthermore, \cite{Lai2023Exponential} shows that the exponential resetting property is satisfied, namely that if there exists $M \ge 0$ such that, for all $k \ge M$, $\phi_k = 0_{p \times n}$, then $\lim_{k \rightarrow \infty} P_k^{-1} = R_\infty$.

Comparing \eqref{eqn: Pinv update exp reset} and \eqref{eqn: theta update exp reset} to \eqref{eqn: GFRLS Pinv Update} and \eqref{eqn: GFRLS theta Update}, it follows that exponential resetting is a special case of GF-RLS where, for all $k \ge 0$, $\Gamma_k = I_p$ and 
\begin{align}
    F_k = (1-\lambda)(P_k^{-1} - R_\infty).
\end{align}
Note that, for all $k \ge 0$, $P_k^{-1} - F_k = \lambda P_k^{-1} + (1-\lambda)R_\infty \succ 0$.
Furthermore, Proposition 6 of \cite{Lai2023Exponential} shows that, for all $k \ge 0$, $P_k^{-1} \succeq \lambda^k P_0^{-1} + (1-\lambda^k)R_\infty$. 
Note that if $P_0^{-1} \succeq R_\infty$, then, for all $k \ge 0$, $P_k^{-1} \succeq \lambda^k R_\infty + (1-\lambda^k)R_\infty = R_\infty$ implying that $F_k \succeq 0$. 
Therefore, if $P_0^{-1} \succeq R_\infty$, then exponential resetting is proper.


\subsection{Covariance Resetting}
A simple ad-hoc extension of RLS is \textit{covariance resetting} \cite{goodwin1983deterministic} where, if a criterion for resetting is met at step $k$, then the covariance matrix $P_k$ is reset to a desired positive-definite matrix, $P_{\infty,k} \in \BBR^{n \times n}$. Covariance resetting gives, for all $k \ge 0$, the update equations
\begin{align}
\label{eqn: covariance resetting Pinv update}
    P_{k+1}^{-1} &= \begin{cases}
        P_{\infty,k}^{-1} + \phi_k^\rmT \phi_k & \textnormal{criterion is met}, \\
        P_k^{-1} + \phi_k^\rmT \phi_k & \textnormal{otherwise},
    \end{cases} \\
    \theta_{k+1} &= \theta_k + P_{k+1} \phi_k^\rmT (y_k - \phi_k \theta_k). \label{eqn: covariance resetting theta update}
\end{align}

Comparing \eqref{eqn: covariance resetting Pinv update} and \eqref{eqn: covariance resetting theta update} to \eqref{eqn: GFRLS Pinv Update} and \eqref{eqn: GFRLS theta Update}, it follows that covariance resetting is a special case of GF-RLS where, for all $k \ge 0$, $\Gamma_k = I_p$ and 
\begin{align}
    F_k = \begin{cases}
        P_k^{-1} - P_{\infty,k}^{-1} & \textnormal{criterion is met}, \\
        0_{n \times n} & \textnormal{otherwise}.
    \end{cases}
\end{align}
Note that, for all $k \ge 0$, 
\begin{align}
    P_k^{-1} - F_k = 
    \begin{cases}
        P_{\infty,k}^{-1} & \textnormal{criterion is met}, \\
        P_k^{-1} & \textnormal{otherwise},
    \end{cases}
\end{align}
and hence $P_k^{-1} - F_k \succeq 0$. 
Moreover, note that when a criterion for resetting is met, $F_k \succeq 0$ if and only of $P_k \preceq P_{\infty,k}$. Thus, if $P_k \preceq P_{\infty,k}$ whenever a criterion for resetting is met, then covariance resetting is proper.
 
%
Covariance resetting can similarly be applied to any RLS extension, resetting the covariance when a criterion is met, and following the nominal algorithm otherwise. Such an algorithm would also be a special case of GF-RLS. 

\subsection{Directional Forgetting by Information Matrix Decomposition}
%

A directional forgetting algorithm based on the decomposition of the information matrix (i.e. inverse of the covariance matrix) is presented in \cite{cao2000directional}.
%
%
This method was developed in the special case of scalar measurements ($p=1$) and can be summarized by the update equations
\begin{align}
    R_{k+1} &= \bar{R}_k + \phi_k^\rmT \phi_k, \label{eqn: R_k update dir forget} \\
    P_{k+1} &= \bar{P}_k - \frac{\bar{P}_k \phi_k^\rmT  \phi_k \bar{P}_k}{1 + \phi_k \bar{P}_k \phi_k^\rmT}, \label{eqn: P_k update dir forget} \\
    \theta_{k+1} &= \theta_k + P_{k+1} \phi_k^\rmT(y_k-\phi_k \theta_k),  \label{eqn: theta_k update dir forget}
\end{align}
where
\begin{align}
    \bar{R}_k &\triangleq \begin{cases}
        R_k - (1-\lambda) \frac{R_k \phi_k^\rmT  \phi_k R_k}{\phi_k R_k \phi_k^\rmT} & \Vert \phi_k \Vert > \varepsilon, \\
        R_k & \Vert \phi_k \Vert \le \varepsilon,
    \end{cases} 
    \\
    \bar{P}_k &\triangleq \begin{cases}
        P_k + \frac{1-\lambda}{\lambda} \frac{\phi_k^\rmT  \phi_k}{\phi_k R_k \phi_k^\rmT} & \Vert \phi_k \Vert > \varepsilon, \\
        P_k & \Vert \phi_k \Vert \le \varepsilon,
    \end{cases} 
\end{align}
and where $\varepsilon > 0$, $0 < \lambda \le 1$ and, for all $k \ge 0$, $R_k = P_k^{-1}$ and $\bar{R}_k = \bar{P}_k^{-1}$.

Comparing \eqref{eqn: R_k update dir forget} and \eqref{eqn: theta_k update dir forget} to \eqref{eqn: GFRLS Pinv Update} and \eqref{eqn: GFRLS theta Update}, it follows that directional forgetting by information matrix decomposition is a special case of GF-RLS where, for all $k \ge 0$, $\Gamma_k = I_p$ and $F_k = R_k - \bar{R}_k$. If $\Vert \phi_k \Vert > \varepsilon$, then $F_k$ can be expressed
\begin{align}
    F_k = (1-\lambda) \frac{P_k^{-1} \phi_k^\rmT  \phi_k P_k^{-1}}{\phi_k P_k^{-1} \phi_k^\rmT},
\end{align}
otherwise, $F_k = 0_{n \times n}$.
It is shown in \cite{cao2000directional} that, for all $k \ge 0$, $\bar{R}_k = P_k^{-1} - F_k \succ 0$ and $F_k \succeq 0$. Therefore, directional forgetting by information matrix decomposition is proper.


\subsection{Variable-Direction Forgetting}
Variable-direction forgetting was developed in \cite{goel2020recursive} and is based on the singular value decomposition of the inverse covariance matrix $P_k^{-1}$.
%
%
For all $k \ge 0$, a positive-definite $\Lambda_k \in \BBR^{n \times n}$ is constructed for the update equations
\begin{align}
    P_{k+1}^{-1} &= \Lambda_k P_{k}^{-1} \Lambda_k + \phi_k^\rmT \phi_k, \label{eqn: Pinv update variable-direction} \\
    \theta_{k+1} &= \theta_k + P_{k+1} \phi_k^\rmT(y_k-\phi_k \theta_k). \label{eqn: theta update variable-direction}
\end{align}
Details on constructing $\Lambda_k$ can be found in equations (67) and (68) of \cite{goel2020recursive}. 
%
%
Comparing \eqref{eqn: Pinv update variable-direction} and \eqref{eqn: theta update variable-direction} to \eqref{eqn: GFRLS Pinv Update} and \eqref{eqn: GFRLS theta Update}, it follows that variable-direction forgetting is a special case of GF-RLS where, for all $k \ge 0$, $\Gamma_k = I_p$ and 
\begin{align}
    F_k = P_k^{-1} - \Lambda_k P_k^{-1} \Lambda_k.
\end{align}
Note that, for all $k \ge 0$, $P_k^{-1} - F_k = \Lambda_k P_k^{-1} \Lambda_k \succ 0$. Moreover, it is shown in the proof of Proposition 9 of \cite{goel2020recursive} that, for all $k \ge 0$, $P_k^{-1} - \Lambda_k P_k^{-1} \Lambda_k \succeq 0$. Therefore, variable-direction forgetting is proper.


%

\subsection{Tracking of Slowly Varying Parameters by Directional Forgetting}
\label{subsec: slowly varying}

Another directional forgetting method, developed in \cite{kulhavy1984tracking} and analyzed in \cite{bittanti1990convergence}, was designed to track slowly varying parameters. A simulation study of this method can also be found in \cite{bertin1987tracking}. This method was developed in the special case of scalar measurements ($p=1$) and can be summarized by the update equations
\begin{align}
    P_{k+1}^{-1} &= P_k^{-1} + \beta_k \phi_k^\rmT \phi_k, \label{eqn: Pk update slowly varying} \\
    \theta_{k+1} &=\theta_k + \frac{1}{1 + \phi_k P_k \phi_k^\rmT}P_k \phi_k^\rmT(y_k - \phi_k \theta_k), \label{eqn: theta_k update slowly varying} 
\end{align}
where, for all $k \ge 0$,
\begin{align}
    \beta_k &\triangleq \begin{cases}
        \mu - \frac{1-\mu}{\phi_k P_k \phi_k^\rmT} & \phi_k P_k \phi_k^\rmT > 0, \\
        1 & \phi_k P_k \phi_k^\rmT = 0,
    \end{cases} 
\end{align}

and $0 < \mu \le 1$ is the forgetting factor. To show this method is a special case of GF-RLS, first note that, for all $k \ge 0$, \eqref{eqn: matrix inversion lemma 2} of Lemma \ref{lem: matrix inversion lemma} can be used to rewrite \eqref{eqn: theta_k update slowly varying} as
\begin{align}
    \theta_{k+1} = \theta_k + (P_k^{-1} + \phi_k^\rmT \phi_k)^{-1} \phi_k(y_k - \phi_k\theta_k). \label{eqn: theta_k update slowly varying 2}
\end{align}
Next, defining $\bar{P}_0 \triangleq P_0$ and, for all $k \ge 0$, $\bar{P}_{k+1}^{-1} \triangleq P_k^{-1} + \phi_k^\rmT \phi_k$. It then follows that, for all $k \ge 0$, \eqref{eqn: theta_k update slowly varying 2} and \eqref{eqn: Pk update slowly varying} can be rewritten as
\begin{align}
    \bar{P}_{k+1}^{-1} &= \bar{P}_k^{-1} - (1-\beta_{k-1})\phi_{k-1}^\rmT \phi_{k-1} + \phi_k^\rmT \phi_k, \label{eqn: eqn: Pinv_k update slowly varying 3}  \\
    \theta_{k+1} &= \theta_k + \bar{P}_{k+1}\phi_k^\rmT (y_k - \phi_k \theta_k), \label{eqn: eqn: theta_k update slowly varying 3}
\end{align}
where $\beta_{-1} \triangleq 0$ and $\phi_{-1} \triangleq 0_{1 \times n}$. 

Comparing \eqref{eqn: eqn: Pinv_k update slowly varying 3} and \eqref{eqn: eqn: theta_k update slowly varying 3} to \eqref{eqn: GFRLS Pinv Update} and \eqref{eqn: GFRLS theta Update}, it follows that this direction forgetting method is a special case of GF-RLS where, for all $k \ge 0$, $\Gamma_k = I_p$ and $F_k = (1-\beta_{k-1})\phi_{k-1}^\rmT \phi_{k-1}$. If $\phi_k P_k \phi_k^\rmT > 0$, then $F_k$ simplifies to
\begin{align}
    F_k = (1-\mu)\Big(\phi_{k-1}^\rmT \phi_{k-1} + \frac{\phi_{k-1}^\rmT \phi_{k-1}}{\phi_{k-1} P_{k-1} \phi_{k-1}^\rmT}\Big),
\end{align}
otherwise, $F_k = 0_{n \times n}$. It is shown in \cite{bittanti1990convergence} that, for all $k \ge 0$, $P_k \succ 0$. Therefore, for all $k \ge 0$, $\bar{P}_k^{-1} - F_k = \bar{P}_{k+1}^{-1} - \phi_k^\rmT \phi_k = P_k^{-1} \succ 0$. 
Furthermore, $\mu \le 1$, and hence, for all $k \ge 0$, $F_k \succeq 0$. Therefore, this direction forgetting method is proper.

\subsection{Multiple Forgetting}
\label{subsec: multiple forgetting}
Multiple forgetting was developed in \cite{vahidi2005recursive} for the special case $n = 2$ and $p = 1$ to allow for different forgetting factors for the two parameters being estimated. To introduce multiple forgetting, we write, for all $k \ge 0$, $\phi_k \in \BBR^{1 \times 2}$ as
\begin{align}
    \phi_k &= \begin{bmatrix}
        \phi_{1,k} & \phi_{2,k}
    \end{bmatrix}.
\end{align}
Then, multiple forgetting can be summarized by the update equations
\begin{align}
    R_{1,k+1} &= \lambda_{1,k} R_{1,k} + \phi_{1,k}^2, \label{eqn: mult forget 1} \\
    R_{2,k+1} &= \lambda_{2,k} R_{2,k} + \phi_{2,k}^2, \label{eqn: mult forget 2} \\
    \theta_{k+1} &= \theta_k + L_{{\rm new},k} (y_k - \phi_k \theta_k), \label{eqn: mult forget 3}
\end{align}
where, for all $k \ge 0$, $\lambda_{1,k},\lambda_{2,k} \in (0,1]$, $R_{1,k}, R_{2,k} \in (0,\infty)$, and 
\begin{align}
    L_{{\rm new},k} \triangleq \frac{1}{1 + \frac{\phi_{1,k}^2}{\lambda_{1,k} R_{1,k}} + \frac{\phi_{2,k}^2}{\lambda_{2,k} R_{2,k}}  }
    \begin{bmatrix}
        \frac{\phi_{1,k}}{\lambda_{1,k} R_{1,k}} \\ \frac{\phi_{2,k}}{\lambda_{2,k} R_{2,k}}
    \end{bmatrix}.
    \label{eqn: mult forget 4}
\end{align}
It was further shown in \cite{fraccaroli2015new} that \eqref{eqn: mult forget 1} through \eqref{eqn: mult forget 4} are equivalent to the update equations
\begin{align}
    R_{k+1} &= \begin{bmatrix}
        \lambda_{1,k} R_{1,k} & 0 \\
        0 & \lambda_{2,k} R_{2,k}
    \end{bmatrix} + \phi_k^\rmT \phi_k, \label{eqn: Pinv update multiple forgetting} \\
    \theta_{k+1} &= \theta_k + P_{k+1} \phi_k^\rmT(y_k-\phi_k \theta_k),  \label{eqn: theta update multiple forgetting}
\end{align}
where, for all $k \ge 0$, $R_k \in \BBR^{2 \times 2}$ is positive definite and $P_k \triangleq R_k^{-1} \in \BBR^{2 \times 2}$. Furthermore, for all $k \ge 0$, denote $R_k$ as
\begin{align}
    R_{k} \triangleq \begin{bmatrix}
        R_{1,k} & R_{12,k} \\
        R_{12,k} & R_{2,k}
    \end{bmatrix}. 
    %
\end{align}

Note that \eqref{eqn: Pinv update multiple forgetting} and \eqref{eqn: theta update multiple forgetting} are equivalent to the GF-RLS update equations \eqref{eqn: GFRLS Pinv Update} and \eqref{eqn: GFRLS theta Update} where, for all $k \ge 0$, $\Gamma_k = I_p$ and 
\begin{align}
    F_k = \begin{bmatrix}
        (1-\lambda_{1,k}) R_{1,k} & R_{12,k} \\
        R_{12,k} & (1-\lambda_{2,k}) R_{2,k}
    \end{bmatrix}.
\end{align}
Furthermore, note that, for all $k \ge 0$,
\begin{align}
    P_k^{-1} - F_k 
    = \begin{bmatrix}
        \lambda_{1,k} R_{1,k} & 0 \\
        0 & \lambda_{2,k} R_{2,k}
    \end{bmatrix}
    \succ 0,
\end{align}
since the diagonal elements of positive-definite $R_k$ are positive. 
Note that, for all $k \ge 0$, $F_k$ is not necessarily positive semidefinite.
However, for all $k \ge 0$, since $R_k$ is positive definite, there exist $\lambda_{1,k},\lambda_{2,k} \in (0,1]$ small enough such that $F_k$ is positive semidefinite. 
Hence, if, for all $k \ge 0$, $\lambda_{1,k},\lambda_{2,k}$ are chosen sufficiently small, then multiple forgetting is proper.

\section{Conclusion}

This article develops GF-RLS, a general framework for RLS extensions derived from minimizing a least-squares cost function. Several RLS extensions are shown to be special cases of GF-RLS, and hence, can be derived from the GF-RLS cost function. 
It is important to note that while the update equations of an RLS extension may not, at face value, seem to be a special case of the GF-RLS update equations, they may still be a special case with some re-definitions. For example, see subsections \ref{subsec: data dependent updating}, \ref{subsec: slowly varying}, and \ref{subsec: multiple forgetting}.
This connects a cost function to many RLS extensions which were originally developed as ad-hod modifications to the RLS update equations (e.g. \cite{cao2000directional,kulhavy1984tracking,bittanti1990convergence,
Salgado1988modified,Lai2023Exponential,vahidi2005recursive,
paleologu2008robust,johnstone1982exponential,Dasgupta1987Asymptotically}).

Further, stability and robustness guarantees are presented for GF-RLS. These guarantees facilitate stability and robustness analysis for various RLS extension that are a special cases of GF-RLS.
Furthermore, a specialization of the robustness result presented gives a bound to the asymptotic bias of the least squares estimator in the errors-in-variables problem.
Applications of this analysis include RLS-based adaptive control \cite{nguyen2021predictive,islam2021data} and online transfer function identification \cite{muller1997iterative,akers1997armarkov}.
Similar analysis may be used to derive tighter bounds if specialized to a single extension of RLS.

A practical use of this work is that conditions A1), A2), and A3) provide a general guideline for the future design of RLS extensions, backed by theoretical analysis.
Satisfying conditions A1), A2), and A3) is entirely dependent on the design of the RLS algorithm while conditions A4) through A8) only concern the data being collected.
Hence, if an extension of RLS is designed to satisfy A1), A2), and A3), then the theoretical guarantees of Theorem \ref{theo: GFRLS Lyapunov stability v2} and Theorem \ref{theo: GFRLS UUB} follow under the assumption of conditions A4) through A8).


\section*{Appendix A: Useful Lemmas}


\begin{lem}
\label{lem: quadratic cost minimizer}
Let $A \in \BBR^{n \times n}$ be positive definite, let $b \in \BBR^n$ and $c \in \BBR$, and define $f\colon\BBR^n \rightarrow \BBR$ by $ f(x) \triangleq x^\rmT A x + 2 b^\rmT x + c$.
%
%
Then, $f$ has a unique global minimizer given by $\argmin_{x \in \BBR^n} f(x) = -A^{-1} b.$
\end{lem}

\begin{lem}[Matrix Inversion Lemma]
\label{lem: matrix inversion lemma}
Let $A \in \BBR^{n \times n}$, $U \in \BBR^{n \times p}$, $C \in \BBR^{p \times p}$, and $V \in \BBR^{p \times n}$. If $A$, $C$, and $A+UCV$ are nonsingular, then $C^{-1} + VA^{-1} U$ is nonsingular, and 
\begin{gather}
    (A+UCV)^{-1} = A^{-1} - A^{-1}U(C^{-1} + VA^{-1} U)^{-1} V A^{-1}, \label{eqn: matrix inversion lemma} \\
    (A+UCV)^{-1} U C = A^{-1} U(C^{-1} + V A^{-1} U)^{-1}. \label{eqn: matrix inversion lemma 2}
\end{gather}
\end{lem}

\begin{lem}
\label{lem: block matrix max singular value squared}
    Let $A \in \BBR^{n \times m}$ be the partitioned matrix 
    \begin{align}
        A \triangleq \begin{bmatrix}
            A_{11} & \hdots & A_{1l} \\
            \vdots & \ddots & \vdots \\
            A_{k1} & \hdots & A_{kl}
        \end{bmatrix},
    \end{align}
    where, for all $i\in \{ 1, \hdots ,k \}$ and $j \in \{1, \hdots, l \}$,  $A_{ij} \in \BBR^{n_i \times m_j}$. Then,
    \begin{align}
        \svdmax(A)^2 \le \sum_{i=1}^k \sum_{j=1}^l \svdmax(A_{ij})^2. 
    \end{align}
\end{lem}
\begin{proof}
    See Theorem 1 of \cite{bhatia1990norm}.
\end{proof}

\section*{Appendix B: Discrete-Time Stability Theory}

Let $f\colon\BBN_0 \times \BBR^n\to\BBR^n$ and consider the system
\begin{align}
    x_{k+1} = f(k,x_k),
    \quad
    \label{eq:NonLinSys}
\end{align}
where, for all $k \ge 0$, $x_k \in \BBR^n$, and $f(k,\cdot)$ is continuous.

\begin{defin}
    \label{defn: equilibrium}
    For $x_{\rm eq} \in \BBR^n$, $x_k \equiv x_{\rm eq}$ is an \textit{equilibrium} of system \eqref{eq:NonLinSys} if, for all $k \ge 0$, $f(k,0) = 0$.
\end{defin}

The following definition is given by Definition 13.7 in \cite[pp. 783, 784]{Haddad2008}.

\begin{defin}
\label{defin: Lyaponov Stabilities}
If $x_k \equiv 0$ is an equilibrium of \eqref{eq:NonLinSys}, then define the following:
\begin{enumerate}[leftmargin=*]
    \item[\textit{i)}] The equilibrium $x_k \equiv 0$ of \eqref{eq:NonLinSys} is \textit{Lyapunov stable} if, 
    for all $\varepsilon > 0$, $k_0  \ge 0$, and $x_{k_0} \in \BBR^n$, 
    there exists $\delta > 0$ such that, if  
    $\| x_{k_0} \| < \delta$, then, for all $k \ge k_0$, $\| x_k \| < \varepsilon$.
    \item[\textit{ii)}] The equilibrium $x_k \equiv 0$ of \eqref{eq:NonLinSys} is \textit{uniformly Lyapunov stable} if, 
    for all $\varepsilon > 0$, 
    there exists $\delta > 0$ such that, for all $k_0  \ge 0$ and $x_{k_0} \in \BBR^n$,
    if $\| x_{k_0} \| < \delta$, then, for all $k \ge k_0$, $\| x_k \| < \varepsilon$.
    \item[\textit{iii)}] The equilibrium $x_k \equiv 0$ of \eqref{eq:NonLinSys} is \textit{globally asymptotically  stable} if it is Lyapunov stable and, for all $k_0  \ge 0$ and $x_{k_0} \in \BBR^n$, $\lim_{k \to \infty} x_k = 0$.
    \item[\textit{iv)}] The equilibrium $x_k \equiv 0$ of \eqref{eq:NonLinSys} is \textit{globally uniformly exponentially stable} if there exist $\alpha > 0$ and $\beta>1$ such that, 
    for all $k_0  \ge 0$, $x_{k_0}  \in \BBR^n$, and $k \ge k_0$, $\| x_k \| \le \alpha \| x_{k_0} \| \beta^{-k}$.
\end{enumerate}
\end{defin}

The following result is a specialization of Theorem 13.11 given in \cite[pp. 784, 785]{Haddad2008}.

\begin{theo} 
\label{theo: lyapunov stability}
Let $ \SD \subset \BBR^n$ be an open set such that $0 \in \SD$ and let $x_k \equiv 0$ be an equilibrium of \eqref{eq:NonLinSys}.
Furthermore, let $V\colon \BBN_0 \times \BBR \to \BBR$ and assume that, for all $k \in \BBN_0$, $V(k,\cdot)$ is continuous.
Then the following statements hold:

\begin{enumerate}[leftmargin=*]
    \myitem{\textit{i)}}\label{item: theo lyp stability} 
    If there exists $\alpha > 0$ such that, for all $k \ge 0$ and $x \in \SD$,
    \begin{gather}
        V(k,0) = 0, \label{eq:LS_Cond1} \\
        \alpha \Vert x \Vert^2 \le  V(k,x), \label{eq:LS_Cond2} \\
        V(k+1, f(k,x)) - V(k,x) \le 0, \label{eq:LS_Cond3} 
    \end{gather}
    then the equilibrium $x_k \equiv 0$ of \eqref{eq:NonLinSys} is {Lyapunov stable}.

    \myitem{\textit{ii)}}\label{item: theo unif lyp stability} 
    If there exist $\alpha > 0$, and $\beta > 0$ such that, for all $k \ge 0$ and $x \in \SD$, \eqref{eq:LS_Cond2}, \eqref{eq:LS_Cond3}, and
    \begin{gather}
        V(k,x) \le \beta \Vert x \Vert ^2, \label{eq:ULS_Cond1}  
    \end{gather}
    then the equilibrium $x_k \equiv 0$ of \eqref{eq:NonLinSys} is {uniformly Lyapunov stable}.

    \myitem{\textit{iii)}}\label{item: theo asym stability} 
    If there exist $\alpha > 0$, and $\gamma > 0$ such that, for all $k \ge 0$ and $x \in \BBR^n$, \eqref{eq:LS_Cond1}, \eqref{eq:LS_Cond2} and
    \begin{gather}
        V(k+1, f(k,x)) - V(k,x) \le -\gamma \Vert x \Vert^2, \label{eq:ALS_Cond3}  
    \end{gather}
    then the equilibrium $x_k \equiv 0$ of \eqref{eq:NonLinSys} is {globally asymptotically stable}.
    \myitem{\textit{iv)}}\label{item: theo geo stability} 
    If there exist $\alpha > 0$, $\beta > 0$, $\gamma > 0$ such that, for all $k \ge 0$ and $x \in \BBR^n$, \eqref{eq:LS_Cond2}, \eqref{eq:ULS_Cond1}, and \eqref{eq:ALS_Cond3},
    then the equilibrium $x_k \equiv 0$ of \eqref{eq:NonLinSys} is {globally uniformly exponentially stable}.
\end{enumerate}
\end{theo}

The following definition is given by 
Definition 13.9 in \cite[pp. 789, 790]{Haddad2008}.

\begin{defin}
\label{defin: UUB}
    The system \eqref{eq:NonLinSys} is \textit{globally uniformly ultimately bounded with bound $\varepsilon$} if, for all $\delta \in (0,\infty)$, there exists $K>0$ such that, for all $k_0 \ge 0$ and $x_{k_0} \in \BBR^n$, 
    if $\Vert x_{k_0} \Vert < \delta$, then, for all $k \ge k_0 + K$, $\Vert x_k \Vert < \varepsilon$.
\end{defin}

The following result is a specialization of Corollary 13.5 given in \cite[pp. 790, 791]{Haddad2008}.

\begin{theo}
\label{theo: UUB}
    Let $V\colon \BBN_0 \times \BBR^n \to \BBR$ and assume that, for all $k \in \BBN_0$, $V(k,\cdot)$ is continuous.
    Furthermore, assume that, for all $k \ge 0$ and $x \in \BBR^n$,
    \begin{align}
        \alpha \Vert x \Vert^2 \le  V(k,x) \le \beta \Vert x \Vert ^2.
    \end{align}
    Furthermore, assume there exist $\mu > 0$ and a continuous function $W\colon\BBR^n \rightarrow \BBR$ such that, for all $k \ge 0$ and $\Vert x \Vert > \mu$, $W(x) > 0$ and
    \begin{align}
        V(k+1, f(k,x)) - V(k,x) \le -W(x).
    \end{align}
    Finally, assume that $\sup_{(k,x) \in \BBN_0 \times \bar{\mathcal{B}}_{\mu}(0)} V(k+1,f(k,x))$ exists, where $\bar{\mathcal{B}}_{\mu}(0) \triangleq \{x \in \BBR^n\colon \Vert x \Vert \le \mu \}$. Then, for all $\varepsilon$ such that\footnote{Note that Corollary 13.5 of \cite{Haddad2008} writes $\sup_{(k,x) \in \cdots} V(k,f(k,x))$ which is a typo that has been verified with the author W. M. Haddad of \cite{Haddad2008}.}
    \begin{align}
        \varepsilon \ge \max \Big\{\mu , \sqrt{\sup_{(k,x) \in \BBN_0 \times \bar{\mathcal{B}}_{\mu}(0)} V(k+1,f(k,x))} \Big\},
    \end{align}
    the system \eqref{eq:NonLinSys} is globally uniformly ultimately bounded with bound $\varepsilon$. 
    %
\end{theo}

Next, $k \ge 0$, define $f_k \colon \BBR^n \rightarrow \BBR^n$ by, for all $x \in \BBR^n$,
\begin{align}
    f_k(x) = f(k,x).
\end{align}
Furthermore, let $N \ge 1$ and, for all $l = 0,1,\hdots,N-1$, define $f_l^N\colon \BBN_0 \times \BBR^n\to\BBR^n$ by, for all $j \ge 0$ and $x \in \BBR^n$,
\begin{align}
    f_l^N(j,x) \triangleq (f_{jN+l+N-1} \circ \cdots \circ f_{jN+l+1} \circ f_{jN+l})(x),
\end{align}
and note that, for all $j \ge 0$, $f_l^N(j,\cdot)$ is continuous. Also note that, for all $j \ge 0$,
\begin{align}
    x_{(j+1)N + l} = f^N_l(j,x_{jN+l}).
\end{align}
In other words, $f^N_l$ can be used to evolve the states of \eqref{eq:NonLinSys} at time steps $\{l,N+l,2N+l,\hdots\}$.
Finally, for all $l = 0,1,\hdots,N-1$ and $j \ge 0$, define $x^N_{l,j} \in \BBR^n$ by 
\begin{align}
    x^N_{l,j} \triangleq x_{jN+l},
\end{align}
which gives the system
\begin{align}
    \label{eqn: nonlinsys N step}
    x^N_{l,j+1} = f_l^N(j,x^N_{l,j}).
\end{align}

\begin{lem}
\label{lem: N step UUB}
    Let $N \ge 1$, and assume that, for all $l = 0,\hdots,N-1$, the system \eqref{eqn: nonlinsys N step} is globally uniformly ultimately bounded with bound $\varepsilon$.
    Then, \eqref{eq:NonLinSys} is globally uniformly ultimately bounded with bound $\varepsilon$.
\end{lem}

\begin{proof}
    Let $\delta_0 \in (0,\infty)$, let $k_0 \ge 0$, and let $x_{k_0} \in \BBR^n$. Assume that $\Vert x_{k_0} \Vert < \delta_0$. 
    Note that there exist $j_0 \ge 0$ and $l_0 \in \{0,\hdots,N-1\}$ such that $k_0 \triangleq j_0 N + l_0$, and it follows from assumption that the system $x^N_{l_0,j+1} = f^N(j,x^N_{l_0,j})$ is globally uniformly ultimately bounded with bound $\varepsilon$. 
    Hence, there exists $J_0 \ge 0$ such that, for all $j \ge j_0 + J_0$, $\Vert x_{jN + l_0} \Vert < \varepsilon$.
    Equivalently, for all $j \ge J_0$, $\Vert x_{k_0 + jN} \Vert < \varepsilon$.

    Next, for all $i = 1,2,\hdots,N-1$, note that there exists $\delta_i \in (0,\infty)$ such that $\Vert x_{k_0+i} \Vert < \delta_i$. 
    By similar reasoning as before, there exists $J_i \ge 0$ such that, for all $j \ge J_i$, $\Vert x_{k_0 + i + jN} \Vert < \varepsilon$.

    Finally, let $K \triangleq (\max \{ J_0, \hdots , J_{N-1}\} +1 )N$. Note that, for all $k \ge k_0 + K$, there exist $i \in \{0,1,\hdots,N-1\}$ and $j \ge J_i$ such that $k = k_0 + i+ jN$, and hence $\Vert x_k \Vert < \varepsilon$.
\end{proof}

\section*{Appendix C: Proof of Theorem \ref{theo: GFRLS}}

\begin{proof}[Proof of Theorem \ref{theo: GFRLS}]
First note that it follows from \eqref{eqn: GFRLS Cost} that $J_0(\hat{\theta})$ can be written as $J_0(\hat{\theta}) = \hat{\theta}^\rmT H_0 \hat{\theta} + 2b_0^\rmT \hat{\theta} + c_0$,
where
\begin{align*}
    H_0 &\triangleq \phi_0^\rmT \Gamma_0^{-1} \phi_0 + P_0^{-1} - F_0, \\
    b_0 &\triangleq -\phi_0^\rmT \Gamma_0^{-1} y_0 - (P_0^{-1} - F_0) \theta_0,  \\
    c_0 &\triangleq y_0^\rmT \Gamma_0^{-1} y_0 + \theta_0^\rmT (P_0^{-1}-F_0) \theta_0.
\end{align*}
Defining $P_1 \triangleq H_0^{-1}$, it follows that \eqref{eqn: GFRLS Pinv Update} holds for $k = 0$.
Furthermore, it follows from \eqref{eqn: Pinv - F is pos def GFRLS} with $k = 0$ that $P_0^{-1} - F_k \succ 0$, and hence $H_0 \succeq P_0^{-1} - F_k \succ 0$.
Therefore, Lemma \ref{lem: quadratic cost minimizer} implies that 
$J_0$ has the unique minimizer $\theta_1 \in \BBR^n$ given by
\begin{align*}
    & \theta_1 = -H_0^{-1}b_0 
    = P_1 [ \phi_0^\rmT \Gamma_0^{-1} y_0 + (P_0^{-1} - F_0) \theta_0  ] \\
    &= P_1 [ \phi_0^\rmT \Gamma_0^{-1} y_0 + (P_0^{-1} - F_0 + \phi_0^\rmT \Gamma_0^{-1} \phi_0) \theta_0 - \phi_0^\rmT \Gamma_0^{-1} \phi_0 \theta_0  ] \\
    &=  P_1 [ \phi_0^\rmT \Gamma_0^{-1} y_0 + P_1^{-1} \theta_0 - \phi_0^\rmT \Gamma_0^{-1} \phi_0 \theta_0  ] \\
    &= \theta_0  + P_1 \phi_0^\rmT \Gamma_0^{-1} (y_0 - \phi_0 \theta_0 ).
\end{align*}
Hence, \eqref{eqn: GFRLS theta Update} is satisfied for $k = 0$.

Now, let $k \ge 1$. Note that $J_k(\hat{\theta})$, given by \eqref{eqn: GFRLS Cost}, 
can be written as 
$J_k(\hat{\theta}) = \hat{\theta}^\rmT H_k \hat{\theta} + 2b_k^\rmT \hat{\theta} + c_k$,
where
\begin{align*}
    H_k &=  \sum_{i=0}^k \left(\phi_i^\rmT \Gamma_i^{-1} \phi_i - F_i \right) + P_0^{-1}, \\
    b_k &= \sum_{i=0}^k \left( -\phi_i^\rmT \Gamma_i^{-1} y_{i} + F_i \theta_i \right) - P_0^{-1} \theta_0, \\
    c_k &= \sum_{i=0}^k \left( y_i^\rmT \Gamma_i^{-1} y_i - \theta_i^\rmT F_i \theta_i \right) + \theta_0^\rmT P_0^{-1} \theta_0.
\end{align*}   
Furthermore, $H_k$ and $b_k$ can be written recursively as
\begin{align*}
    H_k = H_{k-1} - F_k + \phi_k^\rmT {\Gamma}_k^{-1} \phi_k, \\
    b_k = b_{k-1} - \phi_k^\rmT {\Gamma}_k^{-1} y_k + F_k \theta_k.
\end{align*}
Defining $P_{k+1} \triangleq H_k^{-1}$, it follows that \eqref{eqn: GFRLS Pinv Update} is satisfied. 
Furthermore, it follows from \eqref{eqn: Pinv - F is pos def GFRLS} that $H_k$ is positive definite. 
Therefore, Lemma \ref{lem: quadratic cost minimizer} implies that $J_k$ has the unique minimizer $\theta_{k+1}$ given by
\begin{align*}
    & \theta_{k+1} = -H_k^{-1} b_k = -P_{k+1}b_k \\
    &= -P_{k+1}(b_{k-1} - \phi_k^\rmT {\Gamma}_k^{-1} y_k + F_k \theta_k) \\
    &= P_{k+1}(P_k^{-1} \theta_k + \phi_k^\rmT {\Gamma}_k^{-1} y_k - F_k \theta_k) \\
    &= P_{k+1} \big[ (P_{k+1}^{-1} - \phi_k^\rmT {\Gamma}_k^{-1} \phi_k + F_k)\theta_k + \phi_k^\rmT {\Gamma}_k^{-1} y_k - F_k \theta_k \big] \\
    &= P_{k+1} \big[ P_{k+1}^{-1} \theta_k - \phi_k^\rmT {\Gamma}_k^{-1} \phi_k \theta_k + \phi_k^\rmT {\Gamma}_k^{-1} y_k \big] \\
    &= \theta_k + P_{k+1} \phi_k^\rmT {\Gamma}_k^{-1} (y_k - \phi_k \theta_k).
\end{align*}
Hence, \eqref{eqn: GFRLS theta Update} is satisfied.
\end{proof}

\section*{Appendix D: Proof of 1) and 2) of Theorem \ref{theo: GFRLS Lyapunov stability v2}}

\begin{lem}
\label{lem: inv(Pinv - Fk) inequality}
    For all $k \ge 0$, let $F_k \succeq 0$ and assume there exists $b \in (0,\infty)$ such that $(P_k^{-1} - F_k)^{-1} \preceq b I_n$.
    Then, $P_k \preceq b I_n$.
\end{lem}
\begin{proof}
    It follows from $(P_k^{-1} - F_k)^{-1} \preceq b I_n$ that $P_k^{-1} - F_k \succeq \frac{1}{b} I_n$, and hence $P_k^{-1} \succeq \frac{1}{b} I_n + F_k \succeq \frac{1}{b} I_n$. Therefore, $P_k \preceq b I_n$. 
\end{proof}

For all $k \ge 0$, define $\Delta V_k \in \BBR^{n \times n}$ by
\begin{align}
\label{eqn: Delta V_k defn}
    \Delta V_k \triangleq - M_k^\rmT P_{k+1}^{-1} M_k + P_k^{-1},
\end{align}
where, for all $k \ge 0$, $M_k$ is defined in \eqref{eqn: M_k defn}.

\begin{lem}
\label{lem: Delta V_k bound}
    For all $k \ge 0$,
    \begin{align}
        \label{eqn: Delta V_k bound}
        \Delta V_k 
        \succeq F_k + \frac{ \bar{\phi}_k^\rmT \bar{\phi}_k}
        {1+ \eigmax(\bar{\phi}_k \bar{\phi}_k^\rmT) \eigmax((P_k^{-1} - F_k)^{-1}) }.
    \end{align}
\end{lem}

\begin{proof}
    Let $k \ge 0$. It follows from substituting \eqref{eqn: M_k v2} into \eqref{eqn: Delta V_k defn} that $\Delta V_k$ can be expanded as
    \begin{align}
        \Delta V_k = & - P_{k+1}^{-1} + 2\bar{\phi}_k^\rmT \bar{\phi}_k  - \bar{\phi}_k^\rmT \bar{\phi}_k P_{k+1} \bar{\phi}_k^\rmT \bar{\phi}_k + P_k^{-1}.
        \label{eqn: DeltaV temp 1}
    \end{align}
    It then follows from substituting \eqref{eqn: GFRLS Pinv Update} into \eqref{eqn: DeltaV temp 1} that
    \begin{align}
        \Delta V_k           
        &= F_k + \bar{\phi}_k^\rmT \bar{\phi}_k - \bar{\phi}_k^\rmT \bar{\phi}_k P_{k+1} \bar{\phi}_k^\rmT \bar{\phi}_k, \nonumber \\
        &=  F_k + \bar{\phi}_k^\rmT ( I_p - \bar{\phi}_k P_{k+1} \bar{\phi}_k^\rmT ) \bar{\phi}_k.
        \label{eqn: DeltaVk temp}
    \end{align}
    Next, define $G_k \in \BBR^{p \times p}$ by 
    \begin{align}
        \label{eqn: G_k defn}
        G_k \triangleq I_p + \bar{\phi}_k (P_k^{-1} - F_k)^{-1} \bar{\phi}_k^\rmT.
    \end{align}
    It follows from \eqref{eqn: Pinv - F is pos def GFRLS} and \eqref{eqn: Pinv - F is pos def GFRLS v2} that $P_k^{-1} - F_k \succ 0_{n \times n}$. 
    Therefore, $\bar{\phi}_k (P_k^{-1} - F_k)^{-1} \bar{\phi}_k^\rmT \succeq 0$. It then follows that $G_k \succeq I_p$ and hence $G_k$ is nonsingular. Next, it follows from substituting \eqref{eqn: GFRLS Pinv Update} into \eqref{eqn: G_k defn} that 
    \begin{align}
        G_k = I_p + \bar{\phi}_k [P_{k+1}^{-1} - \bar{\phi}_k^\rmT \bar{\phi}_k]^{-1} \bar{\phi}_k^\rmT.
        \label{eqn: G_k temp}
    \end{align}
    Finally, applying \eqref{eqn: matrix inversion lemma} of Lemma \ref{lem: matrix inversion lemma} to \eqref{eqn: G_k temp} gives that 
    \begin{align}
        G_k^{-1} = I_p - \bar{\phi}_k P_{k+1} \bar{\phi}_k^\rmT . \label{eqn: deltaV internal term inverse}
    \end{align}
    Substituting \eqref{eqn: deltaV internal term inverse} into \eqref{eqn: DeltaVk temp}, it follows that
    \begin{align}
        \Delta V_k = & F_k  + \bar{\phi}_k^\rmT G_k^{-1} \bar{\phi}_k.
        \label{eqn: DeltaV temp 2}
    \end{align}
    Finally, it follows from \eqref{eqn: G_k defn} that
    \begin{align}
         G_k \preceq \left[ 1+ \eigmax(\bar{\phi}_k^\rmT \bar{\phi}_k) \eigmax\left((P_k^{-1} - F_k)^{-1}\right) \right] I_p.
         \label{eqn: G_k ineq temp}
    \end{align}
    Combining \eqref{eqn: DeltaV temp 2} and \eqref{eqn: G_k ineq temp} yields \eqref{eqn: Delta V_k bound}.
\end{proof}

{\it Proof of statements 1) and 2) of Theorem \ref{theo: GFRLS Lyapunov stability v2}:}
    Define $V \colon \BBN_0 \times \BBR^n \rightarrow \BBR$ by $V(k,\tilde{\theta}) \triangleq \tilde{\theta}^\rmT P_k^{-1} \tilde{\theta}$.
    Note that, for all $k \ge 0$, 
    \begin{align}
        \label{eqn: V(k,0) = 0}
        V(k,0) = 0.
    \end{align}
    Next, from \eqref{eqn: theta tilde update}, 
    we define $f\colon \BBN_0 \times \BBR^n \rightarrow \BBR^n$ by $f(k,\tilde{\theta}) \triangleq M_k \tilde{\theta}$.
    Note that, for all $k \ge 0$ and $\tilde{\theta} \in \BBR^n$, 
    \begin{align*}
        V(k+1,f(k,\tilde{\theta})) - V(k,\tilde{\theta})  =  -\tilde{\theta}^\rmT \Delta V_k \tilde{\theta}.
    \end{align*}
    Then, for all $k \ge 0$ and $\tilde{\theta} \in \BBR^n$, it follows from Lemma \ref{lem: Delta V_k bound} and condition \ref{item: cond F PSD} that 
    \begin{align}
        \label{eqn: V diff le 0}
        V(k+1,f(k,\tilde{\theta})) - V(k,\tilde{\theta}) \le -\tilde{\theta}^\rmT F_k \tilde{\theta} \le 0.
    \end{align}
    We now prove statements \ref{item: stability statement 1} and \ref{item: stability statement 2}:
    \begin{enumerate}[leftmargin=*]
        \item[\ref{item: stability statement 1}] By Lemma \ref{lem: inv(Pinv - Fk) inequality}, conditions \ref{item: cond F PSD} and \ref{item: cond inv(Pkinv - F) upper bound} imply that, for all $k \ge 0$ and $\tilde{\theta} \in \BBR^n$,
        \begin{align}
            \label{eqn: V lower bound}
            \frac{1}{b} \Vert \tilde{\theta} \Vert^2 \le V(k,\tilde{\theta}).
        \end{align}
        Equations \eqref{eqn: V(k,0) = 0}, \eqref{eqn: V diff le 0}, and \eqref{eqn: V lower bound} imply that \eqref{eq:LS_Cond1}, \eqref{eq:LS_Cond2}, and \eqref{eq:LS_Cond3} are satisfied. It then follows from part \ref{item: theo lyp stability} of Theorem \ref{theo: lyapunov stability}  that the equilibrium $\tilde{\theta}_k \equiv 0$ of \eqref{eqn: theta tilde update} is Lyapunov stable.
        \item[\ref{item: stability statement 2}] \ref{item: cond Pk lower bound} further implies that, for all $k \ge 0$ and $\tilde{\theta} \in \BBR^n$,
        \begin{align}
            \label{eqn: V upper bound}
            V(k,\tilde{\theta}) \le \frac{1}{a} \Vert \tilde{\theta} \Vert^2.
        \end{align}
        Equations \eqref{eqn: V diff le 0}, \eqref{eqn: V lower bound}, and \eqref{eqn: V upper bound} imply that \eqref{eq:LS_Cond2}, \eqref{eq:LS_Cond3}, and \eqref{eq:ULS_Cond1} are satisfied. By part \ref{item: theo unif lyp stability} of Theorem \ref{theo: lyapunov stability}, it follows that the equilibrium $\tilde{\theta}_k \equiv 0$ of \eqref{eqn: theta tilde update} is uniformly Lyapunov stable. \hfill {\mbox{$\blacksquare$}}
    \end{enumerate}

\section*{Appendix E: Proof of 3) and 4) of Theorem \ref{theo: GFRLS Lyapunov stability v2}}
For all $k \ge 0$ and $N \ge 1$, define $\Delta^N V_k \in \BBR^{n \times n}$ by
\begin{align}
\label{eqn: DeltaN V_k defn}
    \Delta^N V_k \triangleq - M_k^\rmT \cdots M_{k+N-1}^\rmT P_{k+N}^{-1} M_{k+N-1} \cdots M_{k} + P_k^{-1}.
\end{align}
Next, all $k \ge 0$, $i \ge 1$, define $W_{k,i} \in \BBR^{p \times p}$ by
\begin{align}
    \label{eqn: W_k,(i,j) defn}
    W_{k,i} &\triangleq \bar{\phi}_{k+i} P_{k+1} \bar{\phi}_k^\rmT.
\end{align}
For all $k \ge 0$, define $\Phi_{k,1} \triangleq \bar{\phi}_k$, $\Psi_{k,1} \triangleq \bar{\phi}_k$ and $\mathcal{W}_{k,N} = I_p$. Furthermore, for all $k \ge 0$ and $N \ge 2$, define $\Phi_{k,N} \in \BBR^{Np \times n}$, $\Psi_{k,N} \in \BBR^{Np \times Np}$, and $\mathcal{W}_{k,N} \in \BBR^{Np \times Np}$ by
\begin{align}
    \Phi_{k,N} &\triangleq \begin{bmatrix}
        \bar{\phi}_{k} \\ \vdots \\ \bar{\phi}_{k+N-1}
    \end{bmatrix},
    \\
    \label{eqn: Psi defn}
    \Psi_{k,N} &\triangleq \begin{bmatrix}
        \bar{\phi}_k
        \\
        \bar{\phi}_{k+1} M_k
        \\
        \vdots
        \\
        \bar{\phi}_{k+N-1} M_{k+N-2} \cdots M_{k+1} M_{k}
    \end{bmatrix},
    \\
    \mathcal{W}_{k,N} &\triangleq \begin{bmatrix}
        I_p & 0 & 0 & \cdots & 0  \\
        W_{k,1} & I_p & 0 & \cdots  & 0 \\
        W_{k,2} & W_{k+1,1} & I_p & \cdots  & 0 \\
        \vdots & \vdots & \vdots & \ddots  & \vdots \\
        W_{k,N-1} &  W_{k+1,N-2} &  W_{k+2,N-3} & \cdots & I_p
    \end{bmatrix} .
\end{align}

\begin{lem}
\label{lem: W Psi = Phi}
    For all $k \ge 0$ and $N \ge 1$,
    \begin{align}
    \label{eqn: W Psi = Phi}
        \mathcal{W}_{k,N} 
        \Psi_{k,N}
        =
        \Phi_{k,N},
    \end{align}
\end{lem}

\begin{proof}
    For all $k \ge 0$, and $N = 1$, \eqref{eqn: W Psi = Phi} simplifies to $\bar{\phi}_k = \bar{\phi}_k$. Next, note that, for all $k \ge 0$ and $N \ge 2$, $\mathcal{W}_{k,N} \Psi_{k,N}$ can be written as
    \begin{align}
    \label{eqn: W Psi expanded}
        \mathcal{W}_{k,N} \Psi_{k,N} = \begin{bmatrix}
            \bar{\phi}_k
            \\
            \begin{bmatrix}
                W_{k,1} & I_p
            \end{bmatrix}
            \Psi_{k,2}
            \\
            \begin{bmatrix}
                W_{k,2} & W_{k+1,1} & I_p
            \end{bmatrix}
            \Psi_{k,3}
            \\
            \vdots
            \\
            \begin{bmatrix}
            W_{k,N-1} & 
            \cdots & W_{k+N-2,1} & I_{p}
            \end{bmatrix} 
            \Psi_{k,N}
        \end{bmatrix}.
    \end{align}
    Next, \eqref{eqn: W_k,(i,j) defn} implies that, for all $N \ge 2$ and $0 \le i \le N-2$,
    \begin{align*}
        W_{k+i,N-1-i} \bar{\phi}_{k+i} 
        &= \bar{\phi}_{k+N-1} P_{k+i+1} \bar{\phi}_{k+i}^\rmT \bar{\phi}_{k+i} 
        \\
        &= \bar{\phi}_{k+N-1}(I_p - M_{k+i}) .
    \end{align*}
    Substituting this identity into the left-hand side of \eqref{eqn: W Psi temp identity}, it follows that, for all $N \ge 2$,
    \begin{align*}
        &\begin{bmatrix}
            W_{k,N-1} & W_{k+1,N-2} & \cdots & W_{k+N-2,1}
        \end{bmatrix} \Psi_{k,N-1} 
        \\
        &= W_{k,N-1} \bar{\phi}_k + \sum_{i=0}^{N-2} W_{k+i,N-1-i} \bar{\phi}_{k+i} M_{k+i-1} \cdots M_k
        \\
        &= \bar{\phi}_{k+N-1} \big( I_p - M_k + \sum_{i=1}^{N-2} (I - M_{k+i}) M_{k+i-1} \cdots M_k \big).
    \end{align*}
    Note that this forms a telescoping series, which, by cancellation of successive terms, simplifies to
    \begin{align}
    \label{eqn: W Psi temp identity}
        \begin{bmatrix}
            W_{k,N-1} & W_{k+1,N-2} & \cdots & W_{k+N-2,1}
        \end{bmatrix} 
         \Psi_{k,N-1} \nonumber
        \\
        \hspace{10pt} = \bar{\phi}_{k+N-1}( I_p - M_{k+N-2} \cdots M_{k+1} M_k). 
    \end{align}
    Next, adding $\bar{\phi}_{k+N-1} M_{k+N-2} \cdots M_{k+1} M_k$ to both sides of \eqref{eqn: W Psi temp identity} implies that, for all $N \ge 2$,
    \begin{align}
    \label{eqn: W Psi = Phi prep}
        \begin{bmatrix}
            W_{k,N-1} & W_{k+1,N-2} & \cdots & W_{k+N-2,1} & I_{p}
        \end{bmatrix} 
        \Psi_{k,N} \nonumber
        \\
        =
         \bar{\phi}_{k+N-1}.
    \end{align}
    
    Hence, applying \eqref{eqn: W Psi = Phi prep} to the row partitions of \eqref{eqn: W Psi expanded} yields \eqref{eqn: W Psi = Phi}.
\end{proof}

\begin{lem}
\label{lem: DeltaN V bound}
    Assume that, for all $k \ge 0$, $F_k \succeq 0_{n \times n}$. Also assume there exists $b \in (0,\infty)$ such that, for all $k \ge 0$, $P_k \preceq b I_n$. Furthermore, assume $( \bar{\phi}_k )_{k=0}^\infty$ is persistently exciting with lower bound $\bar{\alpha} > 0$ and persistency window $N$ and bounded with upper bound $\bar{\beta} \in (0,\infty)$.
    Then, for all $k \ge 0$ and $N \ge 1$,
    \begin{align}
        \label{eqn: DeltaN V bound}
        \Delta^N V_k \succeq c_N I_n \succ 0_{n \times n},
    \end{align}
    where
    \begin{align}
        \label{eqn: c_N defn}
        c_N \triangleq \frac{\bar{\alpha}}{N}
        (1+b \bar{\beta})^{-1}
        \left[ 1 + \frac{N-1}{2} \left( b \bar{\beta} \right)^2 \right]^{-1}.
    \end{align}
\end{lem}

\vspace{5pt}

\begin{proof}
    To begin, we show that, for all $k \ge 0$ and $N \ge 1$,
    \begin{align}
        \label{eqn: DeltaN V bound prep} 
        \Delta^N V_k \succeq (1+ b \bar{\beta})^{-1} \Psi_{k,N}^\rmT \Psi_{k,N}.
    \end{align}
    For brevity, we define $\nu \triangleq (1+ b \bar{\beta})^{-1}$. Proof of \eqref{eqn: DeltaN V bound prep} follows by induction on $N$. First, let $k \ge 0$ and consider the base case $N = 1$. Note that $\Delta^1 V_k = \Delta V_k$ and $\Psi_{k,1} = \bar{\phi}_{k}$. 
    Hence, it follows from Lemma \ref{lem: Delta V_k bound} that 
    $\Delta^1 V_k \succeq F_k + \nu \bar{\phi}_k^\rmT \bar{\phi}_k 
    \succeq \nu \Psi_{k,1}^\rmT \Psi_{k,1}$.
    Next, let $N \ge 2$. Note that $\Delta^N V_k$, given by \eqref{eqn: DeltaN V_k defn} can be expressed recursively as
    \begin{align*}
        \Delta^N V_k &= M_k^\rmT (\Delta^{N-1} V_{k+1} - P_{k+1}^{-1}) M_k + P_k^{-1} 
        \\
        &= M_k^\rmT \Delta^{N-1} V_{k+1} M_k + \Delta V_k.
    \end{align*}
    It follows from inductive hypothesis that $\Delta^{N-1} V_{k+1} \succeq \nu \Psi_{k+1,N-1}^\rmT \Psi_{k+1,N-1}$. Substituting into the previous equation gives
    \begin{align*}
        \Delta^N V_k & \succeq \nu M_k^\rmT \Psi_{k+1,N-1}^\rmT \Psi_{k+1,N-1} M_k + \nu \bar{\phi}_k^\rmT \bar{\phi}_k 
        \\
        & = \nu \begin{bmatrix}
             \bar{\phi}_k^\rmT & M_k^\rmT \Psi_{k+1,N-1}^\rmT
        \end{bmatrix}
        \begin{bmatrix}
            \bar{\phi}_k \\ \Psi_{k+1,N-1} M_k
        \end{bmatrix}.
    \end{align*}
    Note that \eqref{eqn: Psi defn} implies that 
    $\begin{bmatrix} \bar{\phi}_k^\rmT & M_k^\rmT \Psi_{k+1,N-1}^\rmT \end{bmatrix} = \Psi_{k,N}^\rmT$ and \eqref{eqn: DeltaN V bound prep} is proven.

    Next, note that, for all $k \ge 0$ and $N \ge 1$, $\mathcal{W}_{k,N}$ is lower triangular with all ones on the main diagonal, and hence is nonsingular. 
    Thus, Lemma \ref{lem: W Psi = Phi} implies that, for all $k \ge 0$ and $N \ge 1$, $\Psi_{k,N} = \mathcal{W}_{k,N}^{-1} \Phi_{k,N}$ and thus 
    \begin{align}
    \label{eqn: DeltaN V bound prep 1}
        \Psi_{k,N}^\rmT \Psi_{k,N} \succeq 
        \frac{\Phi_{k,N}^\rmT \Phi_{k,N}}{\svdmax(\mathcal{W}_{k,N})^2}  .
    \end{align}
    Note that if $N = 1$ then, for all $k \ge 0$, $\SW_{k,1} = I_p$ and hence
    \begin{align}
    \label{eqn: svdmax of W temp 1}
        \svdmax(\mathcal{W}_{k,1})^2 = 1.
    \end{align}
    Next, if $N \ge 2$ then Lemma \ref{lem: block matrix max singular value squared} implies that 
    $\svdmax(\mathcal{W}_{k,N})^2 \le N \svdmax(I_p)^2 + \sum_{i = 1}^{N-1} \sum_{j = 1}^{N-i} \svdmax(W_{k-1+i,j})^2.$
    Note that, for all $1 \le i \le N-1$ and $1 \le j \le N-i$, it follows from \eqref{eqn: W_k,(i,j) defn} that $\svdmax(W_{k-1+i,j}) \le b \bar{\beta}$.
    Using this inequality, it follows that, for all $k \ge 0$ and $N \ge 2$, $\svdmax(\SW_{k,N})^2 \le N + \sum_{i = 1}^{N-1} \sum_{j = 1}^{N-i} \left(b \bar{\beta} \right)^2$, which simplifies to
    \begin{align}
    \label{eqn: svdmax of W temp 2}
        \svdmax(\SW_{k,N})^2 \le N + \frac{N(N-1)}{2} \left(b \bar{\beta} \right)^2.
    \end{align}
    It then follows from \eqref{eqn: svdmax of W temp 1} and \eqref{eqn: svdmax of W temp 2} that, for all $k \ge 0$ and $N \ge 1$,
    \begin{align}
    \label{eqn: DeltaN V bound prep 2}
        \frac{1}{\svdmax(\mathcal{W}_{k,N})^2} \ge \frac{1}{N} \left[ 1 + \frac{N-1}{2} \left( b \bar{\beta} \right)^2 \right]^{-1}.
    \end{align}
    Furthermore, persistent excitation of $( \bar{\phi}_k )_{k=0}^\infty$ implies that, for all $k \ge 0$ and $N \ge 1$,
    \begin{align}
    \label{eqn: DeltaN V bound prep 3}
        \Phi_{k,N}^\rmT \Phi_{k,N} \succeq \bar{\alpha} I_n.
    \end{align}
    Finally, substituting \eqref{eqn: DeltaN V bound prep 1}, \eqref{eqn: DeltaN V bound prep 2}, and \eqref{eqn: DeltaN V bound prep 3} into \eqref{eqn: DeltaN V bound prep} yields \eqref{eqn: DeltaN V bound}.
\end{proof}


{\it Proof of statements 3) and 4) of Theorem \ref{theo: GFRLS Lyapunov stability v2}.}
    Note that repeated substitution of \eqref{eqn: theta tilde update} gives, for all $j \ge 0$, 
    \begin{align*}
        \tilde{\theta}_{(j+1)N} = M_{(j+1)N - 1} \cdots M_{jN} \tilde{\theta}_{jN}.
    \end{align*}
    Hence, we define $f^N \colon \BBN_0 \times \BBR^n \rightarrow \BBR^n$ by, for all $j \ge 0$ and $\tilde{\theta} \in \BBR^n$,
    \begin{align*}
        f^N(j,\tilde{\theta}) \triangleq M_{jN + N - 1} \cdots M_{jN+1} M_{jN}  \tilde{\theta}.
    \end{align*}
    Further, for all $j \ge 0$, define $\tilde{\theta}_{j}^N \in \BBR^n$ by $\tilde{\theta}_{j}^N \triangleq \tilde{\theta}_{jN}$, which yields the system
    \begin{align}
    \label{eqn: proof N step theta tilde update}
        \tilde{\theta}_{j+1}^N = f^N(j,\tilde{\theta}_{j}^N).
    \end{align}
    
    Next, define $V^N \colon \BBN_0 \times \BBR^n \rightarrow \BBR$ by
    \begin{align*}
        V^N(j,\tilde{\theta}) \triangleq \tilde{\theta}^\rmT P_{jN}^{-1} \tilde{\theta},
    \end{align*}
    and note that, for all $j \ge 0$,
    \begin{align}
        \label{eqn: VN(j,0) = 0}
        V^N(j,0) = 0.
    \end{align}
    Also, note that, for all $j \ge 0$,
    \begin{align*}
        V^N\big(j+1,f^N(j,\tilde{\theta})\big) - V^N\big(j,\tilde{\theta}\big) = -\tilde{\theta}^\rmT \Delta^N V_{jN} \tilde{\theta},
    \end{align*}
    where $\Delta^N V_{jN} \in \BBR^{n \times n}$ is defined in \eqref{eqn: DeltaN V_k defn}. It then follows from Lemma \ref{lem: DeltaN V bound} that, for all $j \ge 0$, 
    \begin{align}
        \label{eqn: proof, V N step diff}
        V^N\big(j+1,f^N(j,\tilde{\theta})\big) - V^N\big(j,\tilde{\theta}\big) 
        \le - c_N \Vert \tilde{\theta} \Vert^2,
    \end{align}
    where $c_N > 0$ is defined in \eqref{eqn: c_N defn}.
    Next, it follows from Lemma \ref{lem: inv(Pinv - Fk) inequality} that, for all $k \ge 0$, $P_k \preceq b I_n$ and it follows then from \eqref{eqn: theta tilde update} and \eqref{eqn: M_k v2} that, for all $k \ge 0$,
    \begin{align}
        \Vert \tilde{\theta}_{k+1} \Vert 
        &\le (\svdmax(I_n) + \svdmax(P_{k+1}) \svdmax(\bar{\phi}_k^\rmT \bar{\phi}_k) )\Vert \tilde{\theta}_k \Vert \nonumber
        \\
        &\le (1 + b \bar{\beta} )  \Vert \tilde{\theta}_k \Vert. \label{eqn: norm theta tilde bound}
    \end{align} 
    Hence, \eqref{eqn: norm theta tilde bound} implies that, for all $j \ge 0$ and $l = 1,\hdots,N-1$, 
    \begin{align}
        \label{eqn: proof theta tilde N step bound}
        \Vert \tilde{\theta}_{jN+l} \Vert \le \left(1 + b \bar{\beta} \right)^{N-1} \Vert \tilde{\theta}_{jN} \Vert.
    \end{align}

    We now prove statements \ref{item: stability statement 3} and \ref{item: stability statement 4} of Theorem \ref{theo: GFRLS Lyapunov stability v2}:
    \begin{enumerate}[leftmargin=*]
        \item[\ref{item: stability statement 3}] By Lemma \ref{lem: inv(Pinv - Fk) inequality},  conditions \ref{item: cond F PSD} and \ref{item: cond inv(Pkinv - F) upper bound} implies that, for all $j \ge 0$ and $\tilde{\theta} \in \BBR^n$,
        \begin{align}
            \label{eqn: Vj lower bound}
            \frac{1}{b} \Vert \tilde{\theta} \Vert ^2 \le V^N(j,\tilde{\theta}).
        \end{align}
        Equations \eqref{eqn: VN(j,0) = 0}, \eqref{eqn: proof, V N step diff}, and \eqref{eqn: Vj lower bound} imply that \eqref{eq:LS_Cond1}, \eqref{eq:LS_Cond2}, and \eqref{eq:ALS_Cond3} are satisfied. Hence, by part \ref{item: theo asym stability} of Theorem \ref{theo: lyapunov stability}, the equilibrium $\tilde{\theta}_{j}^N \equiv 0$ of \eqref{eqn: proof N step theta tilde update} is globally asymptotically stable. 
        
        We now show that the equilibrium $\tilde{\theta}_{k} \equiv 0$ of \eqref{eqn: theta tilde update} is also globally asymptotically stable. 
        Let $\varepsilon > 0$ and $k_0 \ge 0$. Write $k_0 = j_0 N + l_0$, where $j_0 \ge 0$, and $0 \le l_0 \le N-1$. Since the equilibrium $\tilde{\theta}_{j}^N \equiv 0$ of \eqref{eqn: proof N step theta tilde update} is Lyapunov stable, we can choose $\delta$ such that $\Vert \tilde{\theta}_{j_0 N} \Vert < \delta$ implies that, for all $j \ge j_0$, 
        \begin{align}
        \label{eqn: GFRLS proof thetatilde_JN inequality}
            \Vert \tilde{\theta}_{jN} \Vert < \varepsilon (1 + b \bar{\beta})^{\frac{1}{N-1}}. 
        \end{align}
        Let $\Vert \tilde{\theta}_{k_0} \Vert < \delta$. For all $k^* \ge k_0$, write $k^* = j^*N + l^*$ and note that $j^* \ge j_0$. It then follows from \eqref{eqn: proof theta tilde N step bound} and \eqref{eqn: GFRLS proof thetatilde_JN inequality} that, for all $k^* \ge k_0$,
        \begin{align*}
            \Vert \tilde{\theta}_{k^*} \Vert 
            &\le (1+b \bar{\beta})^{N-1} \Vert \tilde{\theta}_{j^*N} \Vert 
            \\
            &< (1+b \bar{\beta})^{N-1} \varepsilon (1+b \bar{\beta})^{\frac{1}{N-1}} 
            = \varepsilon.
        \end{align*}
        Hence, the equilibrium $\tilde{\theta}_{k} \equiv 0$ of \eqref{eqn: theta tilde update} is Lyapunov stable. A similar argument using \eqref{eqn: proof theta tilde N step bound} can be used to show that $\lim_{j \rightarrow \infty} \tilde{\theta}_{jN} = 0$ implies that $\lim_{k \rightarrow \infty} \tilde{\theta}_{k} = 0$. Therefore, the equilibrium $\tilde{\theta}_{k} \equiv 0$ of \eqref{eqn: theta tilde update} is globally asymptotically stable.

        \item[\ref{item: stability statement 4}] Condition \ref{item: cond Pk lower bound} further implies that, for all $j \ge 0$ and $\tilde{\theta} \in \BBR^n$, 
        \begin{align}
            \label{eqn: Vj upper bound}
            V^N(j,\tilde{\theta}) \le \frac{1}{a} \Vert \tilde{\theta} \Vert ^2.
        \end{align}
       Equations \eqref{eqn: proof, V N step diff}, \eqref{eqn: Vj lower bound}, and \eqref{eqn: Vj upper bound} imply that \eqref{eq:LS_Cond2}, \eqref{eq:ULS_Cond1}, and \eqref{eq:ALS_Cond3} are satisfied. 
        Hence, by part \ref{item: theo geo stability} of Theorem \ref{theo: lyapunov stability}, the equilibrium $\tilde{\theta}_{j}^N \equiv 0$ of \eqref{eqn: proof N step theta tilde update} is globally uniformly exponentially stable. 
        Using \eqref{eqn: proof theta tilde N step bound} in an argument similar to the proof of statement 3), it can be shown that the equilibrium $\tilde{\theta}_{k} \equiv 0$ of \eqref{eqn: theta tilde update} is also globally uniformly exponentially stable. 
        \hfill {\mbox{$\blacksquare$}}
    \end{enumerate}
    %

\section*{Appendix F: Lemmas used in the proof of Theorem \ref{theo: GFRLS UUB}}

Suppose, for all $k \ge 0$, there exists $\zeta_k \in \BBR^n$ such that, 
\begin{align}
    \label{eqn: thetacheck update assumption}
    \check{\theta}_{k+1} = M_k(\check{\theta}_k - \zeta_k),
\end{align}
where $\check{\theta}_k$ is defined in \eqref{eqn: theta bar defn} and $M_k$ is defined in \eqref{eqn: M_k defn}.
Next, for all $k \ge 0$ and $i \ge 1$, define $\mathcal{M}_{k,i} \in \BBR^{n \times n}$ by
\begin{align}
    \mathcal{M}_{k,i} \triangleq \begin{cases}
        M_k & i = 1, \\
        M_{k+i-1} \cdots M_{k+1} M_k & i \ge 2.
    \end{cases}
\end{align}
Moreover, for all $k \ge 0$, let $P_{k}^{-\frac{1}{2}} \in \BBR^{n \times n}$ be the unique positive-semidefinite matrix such that $P_{k}^{-1} = P_{k}^{-\frac{1}{2}\rmT} P_{k}^{-\frac{1}{2}}$.

\begin{lem}
\label{lem: thetabar N step update}
    Assume, for all $k \ge 0$, there exists $\zeta_k \in \BBR^n$ such that \eqref{eqn: thetacheck update assumption} holds.
    Then, for all $N \ge 1$ and $k \ge 0$, 
    \begin{align}
        \label{eqn: thetabar N step update}
        \check{\theta}_{k+N} = \mathcal{M}_{k,N} \check{\theta}_k - \overline{\mathcal{M}}_{k,N} \bar{\zeta}_{k,N},
    \end{align}
    where $\overline{\mathcal{M}}_{k,N} \in \BBR^{n \times Nn}$ and $\bar{\zeta}_{k,N} \in \BBR^{Nn \times 1}$ are defined
    \begin{align}
    \label{eqn: M overline defn}
        \overline{\mathcal{M}}_{k,N} &\triangleq \begin{bmatrix}
            \mathcal{M}_{k,N} & \mathcal{M}_{k+1,N-1} & \cdots 
            & \mathcal{M}_{k+N-1,1}
        \end{bmatrix}, 
        \\
        \label{eqn: zeta bar defn}
        \bar{\zeta}_{k,N} &\triangleq \begin{bmatrix}
            \zeta_{k}^\rmT & \zeta_{k+1}^\rmT & \cdots & \zeta_{k+N-1}^\rmT
        \end{bmatrix}^\rmT.
    \end{align}
\end{lem}

\begin{proof}
    Let $k \ge 0$ and proof follows by induction on $N \ge 1$. First, consider the base case $N = 1$. Note that $\overline{\mathcal{M}}_{k,1} = \mathcal{M}_{k,1} = M_k$ and $\bar{\zeta}_{k,1} = \zeta_{k}$. 
    Hence, \eqref{eqn: thetabar N step update} follows immediately from \eqref{eqn: thetacheck update assumption}. Next, let $N \ge 2$. By inductive hypothesis, $\check{\theta}_{k+N-1} = \mathcal{M}_{k,N-1} \check{\theta}_k - \overline{\mathcal{M}}_{k,N-1} \bar{\zeta}_{k,N-1}$. 
    Furthermore, it follows from \eqref{eqn: thetacheck update assumption} that $\check{\theta}_{k+N} = M_{k+N-1} (\check{\theta}_{k+N-1} - \zeta_{k+N-1})$. Combining these two equalities gives
    \begin{align*}
        \check{\theta}_{k+N}&  =  M_{k+N-1} \mathcal{M}_{k,N-1} \check{\theta}_k \nonumber
        \\
        &+ 
        \begin{bmatrix}
            M_{k+N-1} \overline{\mathcal{M}}_{k,N-1} & M_{k+N-1}
        \end{bmatrix}
        \begin{bmatrix}
            \bar{\zeta}_{k,N-1} \\ \zeta_{k+N-1}
        \end{bmatrix},
    \end{align*}
    which can be rewritten as \eqref{eqn: thetabar N step update}.
\end{proof}

Note that from \eqref{eqn: thetacheck update assumption} and Lemma \ref{lem: thetabar N step update}, it follows that, for all $l = 0,1,\hdots,N-1$ and $j \ge 0$, 
\begin{align}
    \check{\theta}_{(j+1)N + l} = \mathcal{M}_{jN+l,N} \check{\theta}_{jN+l} - \overline{\mathcal{M}}_{jN+l,N} \bar{\zeta}_{jN+l,N}.
\end{align}
Next, for all $l = 0,1,\hdots,N-1$ and $j \ge 0$, define $\check{\theta}^N_{l,j} \in \BBR^n$ by
\begin{align}
    \check{\theta}^N_{l,j} \triangleq \check{\theta}_{jN + l},
\end{align}
which, for all $l = 0,1,\hdots,N-1$, gives the system
\begin{align}
    \label{eqn: thetabar N step with l}
    \check{\theta}^N_{l,j+1} = \mathcal{M}_{jN+l,N} \check{\theta}^N_{l,j} - \overline{\mathcal{M}}_{jN+l,N} \bar{\zeta}_{jN+l,N}.
\end{align}

\begin{lem}
    \label{lem: GFRLS UUB N step}
    Assume, for all $k \ge 0$, there exists $\zeta_k \in \BBR^n$ such that \eqref{eqn: thetacheck update assumption} holds.
    Furthermore, assume conditions \ref{item: cond F PSD}, \ref{item: cond inv(Pkinv - F) upper bound}, \ref{item: cond Pk lower bound}, and \ref{item: cond weighted regressor PE and bounded} hold and assume there exists $\zeta \ge 0$ such that, for all $k \ge 0$, 
    \begin{align}
    \label{eqn: zeta_k bound}
        \Vert \zeta_k \Vert \le \zeta.
    \end{align}
    Then, for all $l = 0,1,\hdots,N-1$, the system \eqref{eqn: thetabar N step with l} is globally uniformly ultimately bounded with bound $\varepsilon^* \zeta$, where $\varepsilon^*$ is given by \eqref{eqn: epsilonbar star defn}.
\end{lem}

\begin{proof}
    For brevity, we prove the case $l = 0$. The cases $l = 1,\hdots,N-1$ can be shown similarly.

    Define $\check{f}^N \colon \BBN_0 \times \BBR^n \rightarrow \BBR^n$ by
    \begin{align}
        \label{eqn: checkf^N defn}
        \check{f}^N(j,\check{\theta}) \triangleq \mathcal{M}_{jN,N} \check{\theta} - \overline{\mathcal{M}}_{jN,N} \bar{\zeta}_{jN,N},
    \end{align}
    and note that, for all $j \ge 0$, $\check{\theta}^N_{0,j+1} = f^N(j,\check{\theta}^N_{0,j})$. Furthermore, define $V^N \colon \BBN_0 \times \BBR^n \rightarrow \BBR$ by
    \begin{align}
    \label{eqn: V^N defn UUB proof}
        V^N(j,\check{\theta}) \triangleq \check{\theta}^\rmT P_{jN}^{-1} \check{\theta}.
    \end{align}
    Note that, for all $j \ge 0$ and $\check{\theta} \in \BBR^n$,
    \begin{align}
    \label{eqn: V^N bound UUB proof}
        \frac{1}{b} \Vert \check{\theta} \Vert^2 \le V^N(j,\check{\theta}) \le \frac{1}{a} \Vert \check{\theta} \Vert^2, 
    \end{align}
    where the lower bound follows from Lemma \ref{lem: inv(Pinv - Fk) inequality} and conditions \ref{item: cond F PSD} and \ref{item: cond inv(Pkinv - F) upper bound} and where the upper bound follows from condition \ref{item: cond Pk lower bound}. 
    
    Next, by substituting \eqref{eqn: checkf^N defn} into \eqref{eqn: V^N defn UUB proof}, it follows that, for all $j \ge 0$ and $\check{\theta} \in \BBR^n$,
    \begin{align}
        V^N( & j+1,\check{f}^N(j,\check{\theta})) - V^N(j,\check{\theta}) \nonumber
        \\
        = & - \check{\theta}^\rmT \Delta^N V_{jN} \check{\theta} 
        - 2 \check{\theta}^\rmT \mathcal{M}_{jN,N}^\rmT P_{(j+1)N}^{-1} \overline{\mathcal{M}}_{jN,N} \bar{\zeta}_{jN,N} \nonumber
        \\
        & + \bar{\zeta}_{jN,N}^\rmT  \overline{\mathcal{M}}_{jN,N}^\rmT P_{(j+1)N}^{-1} \overline{\mathcal{M}}_{jN,N} \bar{\zeta}_{jN,N},
        \label{eqn: UUB time-varying deltaV 1}
    \end{align}  
    where $\Delta^N V_{jN} \in \BBR^{n \times n}$ is defined in \eqref{eqn: DeltaN V_k defn}.
    
    Next, since conditions \ref{item: cond F PSD}, \ref{item: cond inv(Pkinv - F) upper bound}, \ref{item: cond Pk lower bound}, and \ref{item: cond weighted regressor PE and bounded} hold, it then follows from Lemma \ref{lem: DeltaN V bound} that, for all $k \ge 0$ and $N \ge 1$, 
    \begin{align}
        \label{eqn: theo 3 temp 0}
        -\Delta^N V_{k} = \mathcal{M}_{k,N}^\rmT P_{k+N}^{-1} \mathcal{M}_{k,N} - P_k^{-1} \preceq - c_N I_n,
    \end{align}
    where $c_N > 0$ is defined in \eqref{eqn: c_N defn}.
    It then follows from \eqref{eqn: theo 3 temp 0} that, for all $k \ge 0$ and $N \ge 1$,
    \begin{align}
        \mathcal{M}_{k,N}^\rmT P_{k+N}^{-1} \mathcal{M}_{k,N} 
        &= (P_{k+N}^{-\frac{1}{2}} \mathcal{M}_{k,N})^\rmT (P_{k+N}^{-\frac{1}{2}} \mathcal{M}_{k,N}) \nonumber
        \\
        &\preceq P_k^{-1} -c_N I_n \preceq (\frac{1}{a} - c_N) I_n.
        \label{eqn: theo 3 temp 0.5}
    \end{align}
    Therefore, \eqref{eqn: theo 3 temp 0.5} implies that, for all $k \ge 0$ and $N \ge 1$,
    \begin{align}
    \label{eqn: P^1/2 M max singular value}
        \svdmax(P_{k+N}^{-\frac{1}{2}} \mathcal{M}_{k,N}) \le \sqrt{\frac{1}{a} - c_N}.
    \end{align}

    Next, it follows from \eqref{eqn: M overline defn} that, for all $j \ge 0$,
    \begin{align}
    \label{eqn: theo 3 temp 1}
        & P_{(j+1)N}^{-\frac{1}{2}} \overline{\mathcal{M}}_{jN,N} \nonumber
        \\
        &= 
        \begin{bmatrix}
            P_{(j+1)N}^{-\frac{1}{2}} \mathcal{M}_{jN,N} 
            & \cdots 
            & P_{(j+1)N}^{-\frac{1}{2}} \mathcal{M}_{jN+N-1,1}
        \end{bmatrix}.
    \end{align}
    Applying Lemma \ref{lem: block matrix max singular value squared} to \eqref{eqn: theo 3 temp 1} and substituting the bound \eqref{eqn: P^1/2 M max singular value} then gives
    \begin{align}
    \label{eqn:  P*M-overline inequality 0}
        \svdmax(P_{(j+1)N}^{-\frac{1}{2}} \overline{\mathcal{M}}_{jN,N}) \le \sqrt{\sum\nolimits_{i=1}^N (\frac{1}{a} - c_i)}.
    \end{align}
    Furthermore, it is easy to verify from definition \eqref{eqn: c_N defn} that, for all $i \ge 1$, $c_{i+1} < c_i$. Therefore, for all $j \ge 0$,
    \begin{align}
    \label{eqn:  P*M-overline inequality}
        \svdmax(P_{(j+1)N}^{-\frac{1}{2}} \overline{\mathcal{M}}_{jN,N}) \le \sqrt{N (\frac{1}{a} - c_N)}.
    \end{align}
    Bounds \eqref{eqn: theo 3 temp 0.5} and \eqref{eqn:  P*M-overline inequality} then imply that, for all $j \ge 0$,
    \begin{align}
        \label{eqn: UUB quad ineq 1}
        \mathcal{M}_{jN,N}^\rmT P_{(j+1)N}^{-1} \overline{\mathcal{M}}_{jN,N} 
        %
        %
        \preceq \sqrt{N} (\frac{1}{a} - c_N) I_n,
        \\
        \label{eqn: UUB quad ineq 2}
        \overline{\mathcal{M}}_{jN,N}^\rmT P_{(j+1)N}^{-1} \overline{\mathcal{M}}_{jN,N}
        %
        %
        \preceq N (\frac{1}{a} - c_N) I_n.
    \end{align}
    Finally, it follows from \eqref{eqn: zeta_k bound} and \eqref{eqn: zeta bar defn} that, for all $j \ge 0$,
    \begin{align}
        \label{eqn: deltabar bound}
        \Vert \bar{\zeta}_{jN,N} \Vert \le \sqrt{N} \zeta.
    \end{align}
    Substituting \eqref{eqn: theo 3 temp 0}, \eqref{eqn: UUB quad ineq 1}, \eqref{eqn: UUB quad ineq 2}, and \eqref{eqn: deltabar bound} into \eqref{eqn: UUB time-varying deltaV 1} and applying the Cauchy Schwarz inequality gives, for all $j \ge 0$ and $\check{\theta} \in \BBR^n$, the bound
    \begin{align}
        & V^N (j+1,\check{f}^N(j,\check{\theta})) - V^N(j,\check{\theta}) \nonumber
        \\
        & \le -c_N \Vert \check{\theta} \Vert ^2 + 2 N \zeta (\frac{1}{a} - c_N) \Vert \check{\theta} \Vert + N^2 \zeta^2 (\frac{1}{a} - c_N). \label{eqn: V diff theo 3}
    \end{align}

    Next, note that $\Delta_N \in \BBR$, defined in \eqref{eqn: Delta_N bar defn}, can be written as
    \begin{align}
    \label{eqn: DeltaN identity}
        \Delta_N = \frac{1}{ac_N}-1.
    \end{align}
    It then follows from \eqref{eqn: V diff theo 3} that, for all $j \ge 0$ and $\check{\theta} \in \BBR^n$,
    \begin{align}
        \frac{1}{c_N} & \left[ V^N (j+1,\check{f}^N(j,\check{\theta})) - V^N(j,\check{\theta}) \right] \nonumber
        \\
        & \le - \Vert \check{\theta} \Vert ^2 + 2 N \zeta (\frac{1}{a c_N} - 1) \Vert \check{\theta} \Vert + N^2 \zeta^2 (\frac{1}{a c_N} - 1). \nonumber
        \\
        & = - \Vert \check{\theta} \Vert ^2 + 2 N \zeta \Delta_N \Vert \check{\theta} \Vert + N^2 \zeta^2 \Delta_N.\label{eqn: Vdiff quad eq}
    \end{align} 
    Next, we define $\mu_N \in \BBR$ by 
    \begin{align}
    \label{eqn: mu_N defn}
        \mu_N \triangleq \big(\Delta_N + \sqrt{\Delta_N + \Delta_N^2} \, \big)N \zeta.
    \end{align}
    It follows from solving the quadratic equation \eqref{eqn: Vdiff quad eq} that, for all $j \ge 0$ and $\check{\theta} \in \BBR^n$ such that $\Vert \check{\theta} \Vert > \mu_N$,
    \begin{align}
    \label{eqn: V^N diff negative}
        V^N(j+1,\check{f}^N(j,\check{\theta})) - V^N(j,\check{\theta}) < 0.
    \end{align}
    
    Next, it follows from substituting \eqref{eqn: checkf^N defn} into \eqref{eqn: V^N defn UUB proof} that, for all $j \ge 0$ and $\check{\theta} \in \BBR^n$,
    \begin{align}
        \sqrt{V^N (j+1, \check{f}^N (j,\check{\theta}))} \hspace{-40pt} & \nonumber 
        \\
        =& \Vert \mathcal{M}_{jN,N} \check{\theta} - \overline{\mathcal{M}}_{jN,N} \bar{\zeta}_{jN,N} \Vert _{P_{(j+1)N}^{-1}} \nonumber 
        \\
        =& \Vert P_{(j+1)N}^{-\frac{1}{2}} \mathcal{M}_{jN,N} \check{\theta} - P_{(j+1)N}^{-\frac{1}{2}} \overline{\mathcal{M}}_{jN,N} \bar{\zeta}_{jN,N} \Vert \nonumber 
        \\
        \le &  \svdmax(P_{(j+1)N}^{-\frac{1}{2}} \mathcal{M}_{jN,N}) \Vert \check{\theta} \Vert \nonumber 
        \\
        & + \svdmax(P_{(j+1)N}^{-\frac{1}{2}} \overline{\mathcal{M}}_{jN,N}) \Vert  \bar{\zeta}_{jN,N} \Vert,
        \label{eqn: sqrt V^N diff temp 1}
    \end{align}
    where the last inequality follows from triangle inequality.
    Applying inequalities \eqref{eqn: P^1/2 M max singular value}, \eqref{eqn:  P*M-overline inequality}, and \eqref{eqn: deltabar bound} to \eqref{eqn: sqrt V^N diff temp 1}, it follows that, for all $j \ge 0$ and $\check{\theta} \in \BBR^n$,
    \begin{align}
    \label{eqn: sqrt V^N diff prep}
        \sqrt{V^N(j+1,\check{f}^N(j,\check{\theta}))} \le \sqrt{\frac{1}{a} - c_N} \Vert \check{\theta} \Vert  +\sqrt{\frac{1}{a} - c_N} N \zeta.
    \end{align}
    For brevity, denote
    \begin{align*}
        \sqrt{\sup V} \triangleq \sqrt{\sup_{(j,\check{\theta}) \in \BBN_0 \times \bar{\mathcal{B}}_{\mu_N}(0)} V^N(j+1,\check{f}^N(j,\check{\theta}))}. 
    \end{align*}
    It then follows from \eqref{eqn: sqrt V^N diff prep} that
    \begin{align}
    \label{eqn: supV inequality setup0}
        \sqrt{\sup V} \le \sqrt{\frac{1}{a} - c_N}  \Big(\mu_N    +    N \zeta\Big).
    \end{align}
    Substituting \eqref{eqn: mu_N defn} into \eqref{eqn: supV inequality setup0}, it follows that
    \begin{align}
    \label{eqn: supV inequality setup1}
        \sqrt{\sup V} 
        \le \sqrt{\frac{1}{a} - c_N} \left(1 + \Delta_N + \sqrt{\Delta_N + \Delta_N^2
    } \right) N \zeta.
    \end{align}
    Next, note that \eqref{eqn: DeltaN identity} implies that 
    \begin{align}
    \label{eqn: 1/a - cN}
        \frac{1}{a} - c_N = \frac{1}{a} \frac{\Delta_N}{1 + \Delta_N}.
    \end{align}
    Substituting \eqref{eqn: 1/a - cN} into \eqref{eqn: supV inequality setup1} and simplifying, it follows that
    \begin{align}
    \label{eqn: supV bound UUB proof}
        \sqrt{\sup V}
        \le \frac{1}{\sqrt{a}} \Big(\Delta_N + \sqrt{\Delta_N + \Delta_N^2} \Big) N \zeta.
    \end{align}
    Applying Theorem \ref{theo: UUB}, it follows from \eqref{eqn: V^N bound UUB proof}, \eqref{eqn: mu_N defn}, \eqref{eqn: V^N diff negative}, and \eqref{eqn: supV bound UUB proof} that the system \eqref{eqn: thetabar N step with l} with $l = 0$ is globally uniformly ultimately bounded with bound $\varepsilon^* \zeta$, where $\varepsilon^*$ is given by \eqref{eqn: epsilonbar star defn}. 
\end{proof}

\begin{lem}
    \label{lem: GFRLS UUB}
    Assume, for all $k \ge 0$, there exists $\zeta_k \in \BBR^n$ such that \eqref{eqn: thetacheck update assumption} holds.
    Furthermore, assume conditions \ref{item: cond F PSD}, \ref{item: cond inv(Pkinv - F) upper bound}, \ref{item: cond Pk lower bound}, and \ref{item: cond weighted regressor PE and bounded} hold.
    Finally, assume there exists $\zeta \ge 0$ such that, for all $k \ge 0$, $\Vert \zeta_k \Vert \le \zeta$.
    Then, the system \eqref{eqn: thetacheck update assumption} is globally uniformly ultimately bounded with bound $\varepsilon^* \zeta$, where $\varepsilon^*$ is given by \eqref{eqn: epsilonbar star defn}.
\end{lem}

\begin{proof}
    It follows from Lemma \ref{lem: GFRLS UUB N step} that, for all $l = 0,1,\hdots,N-1$, the system \eqref{eqn: thetabar N step with l} is globally uniformly ultimately bounded with bound $\varepsilon^* \zeta$, where $\varepsilon^*$ is given by \eqref{eqn: epsilonbar star defn}.
    It then follows from Lemma \ref{lem: N step UUB} that the system \eqref{eqn: thetacheck update assumption} is also globally uniformly ultimately bounded with bound $\varepsilon^* \zeta$, where $\varepsilon^*$ is given by \eqref{eqn: epsilonbar star defn}.
\end{proof}

\begin{figure}[ht]
    \centering
    \includegraphics[width = 0.5\textwidth]{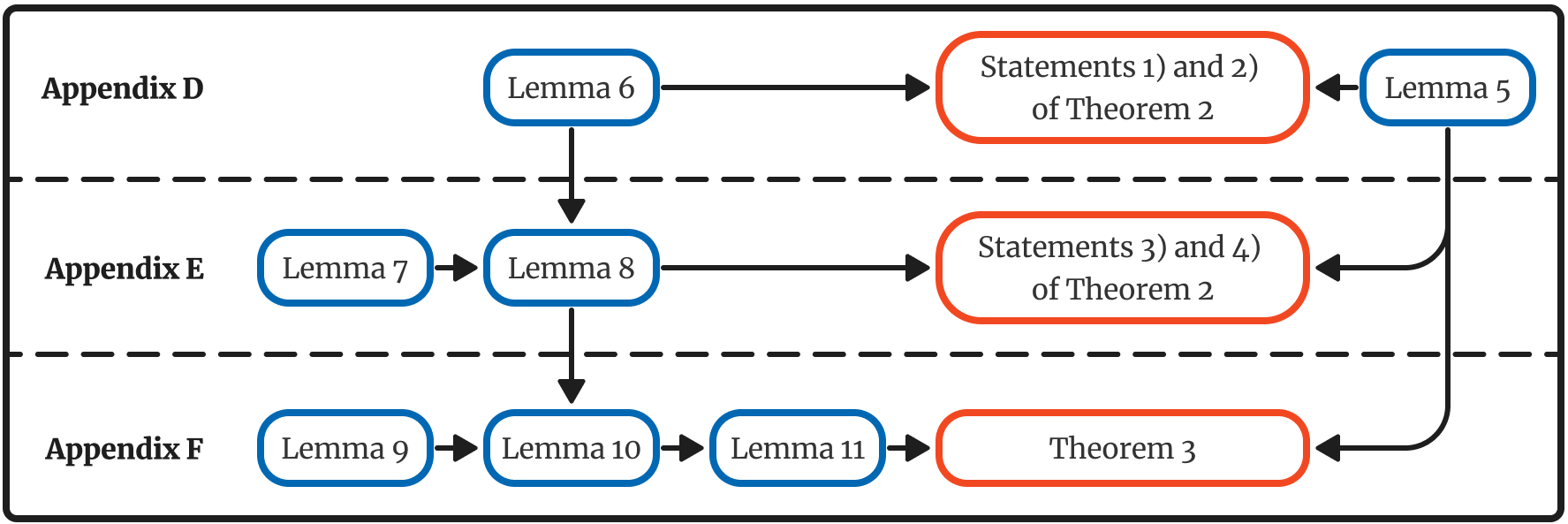}
    \caption{Proof ``roadmap" of Theorem \ref{theo: GFRLS Lyapunov stability v2} and Theorem \ref{theo: GFRLS UUB}.}
    \label{fig:proof-roadmap}
\end{figure}


\bibliographystyle{IEEEtran}
\bibliography{refs}

\begin{IEEEbiography}[{\includegraphics[width=1in,height=1.25in,clip,keepaspectratio]{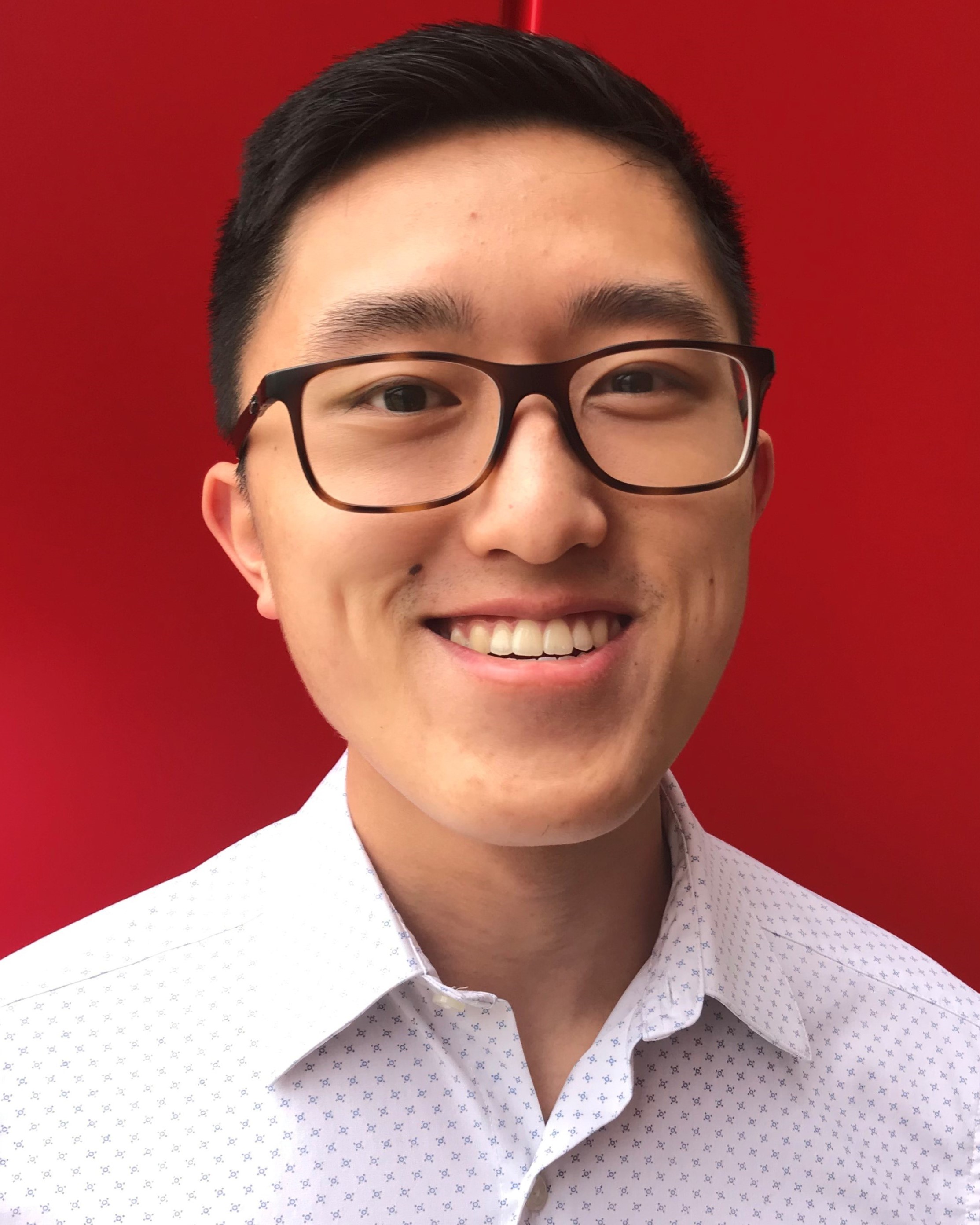}}]{Brian Lai}
received the B.S. degree in mechanical engineering and mathematics from Rutgers University, New Brunswick, NJ, USA, in 2020 
and the M.S. degree in Aerospace Engineering from the University of Michigan, Ann Arbor, MI, USA, in 2022.

He is currently a Ph.D. candidate at the University of Michigan, Ann Arbor, MI, USA.
His research interests include data-driven identification, estimation, and adaptive control for aerospace applications.

Mr. Lai was a recipient of the National Science Foundation Graduate Research Fellowship (GRFP) in 2021.
\end{IEEEbiography}

\begin{IEEEbiography}[{\includegraphics[width=1in,height=1.25in,clip,keepaspectratio]{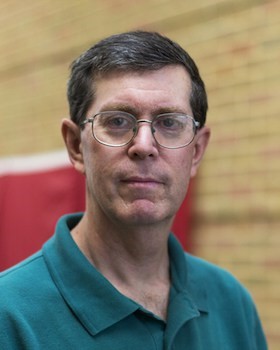}}]{Dennis S. Bernstein} (Life Fellow, IEEE) 
received the Sc.B. degree in applied mathematics from Brown University, Providence, RI, USA, in 1977 
and the M.S.E. and Ph.D. degrees in computer, information, and control engineering from the University of Michigan, Ann Arbor, MI, USA, in 1979 and 1982, respectively.

He is currently the James E. Knott professor in the Aerospace Engineering Department at the University of Michigan in Ann Arbor, Michigan.
His research interests include identification, estimation, and control for aerospace applications.

Prof. Bernstein was Editor-in-Chief of the IEEE Control Systems Magazine from 2003 to 2011 and is the author of \textit{Scalar, Vector, and Matrix Mathematics},
third edition published by Princeton University Press in 2018.

\end{IEEEbiography}




\end{document}